\documentclass[a4paper,11pt]{amsart}
\usepackage{amssymb,dsfont,interval}
\usepackage[left=25mm,top=30mm,bottom=25mm,right=25mm]{geometry}
\usepackage[english]{babel}
\usepackage{hyperxmp}
\usepackage{cmap}
\usepackage{subfig}
\usepackage{graphicx}
\usepackage{mathtools}
\usepackage{float}
\usepackage[pdfdisplaydoctitle=true,
      colorlinks=true,
      urlcolor=blue,
      citecolor=blue,
      linkcolor=blue,
      pdfstartview=FitH,
      pdfpagemode=None,
      bookmarksnumbered=true]{hyperref} 
\usepackage[all]{xy}
\usepackage{mathrsfs}
\usepackage{xfrac}
\usepackage{array}       
\usepackage{booktabs}
\usepackage{enumerate}
\usepackage{todonotes}
\usepackage{changes}
\usepackage{bm}
\usepackage{xcolor}

\newcommand{\red}[1]{{\textcolor{black}{#1}}}

\DeclarePairedDelimiter\floor{\lfloor}{\rfloor}

\graphicspath{{fig/}}

\newtheorem{thm}{Theorem}[section]

\newtheorem{lem}[thm]{Lemma}

\theoremstyle{definition}
\newtheorem{defn}[thm]{Definition}
\theoremstyle{remark}
\newtheorem{rem}[thm]{Remark}

\numberwithin{equation}{section}       

\graphicspath{{fig/}}

\newcommand{\RR}{\mathbb{R}}                
\newcommand{\Pie}{\mathbb{P}\mathrm{iez}}    
\newcommand{\Piez}{\mathcal{P}\mathrm{iez}}    
\newcommand{\Ela}{\mathbb{E}\mathrm{la}}    
\newcommand{\HH}{\mathbb{H}}                

\newcommand{\Sym}{\mathbb{S}} 
  
\newcommand{\SO}{\mathrm{SO}}               
\newcommand{\OO}{\mathrm{O}}                

\newcommand{\ZZ}{\mathbb{Z}}   
\newcommand{\DD}{\mathbb{D}}                
\newcommand{\octa}{\mathbb{O}}              
\newcommand{\ico}{\mathbb{I}}               
\newcommand{\tetra}{\mathbb{T}}             
\newcommand{\triv}{\mathds{1}}              

\newcommand{\fix}{\mathbb{F}\mathrm{ix}}

\newcommand{\vR}{\mathbf{r}}

\newcommand{\uu}{\mathbf{u}}              

\newcommand{\ve}{\mathbf{e}}

\newcommand{\vv}{\mathbf{v}}

\newcommand{\mm}{\mathbf{m}}
 
\newcommand{\ba}{\mathbf{a}}

\newcommand{\bnn}{\pmb{n}}                
\newcommand{\bs}{\mathbf{s}}	

\newcommand{\dT}[1]{\mathrm{\mathbf{#1}}}
\newcommand{\dTs}[1]{\mathrm{\boldsymbol{#1}}}

\newcommand{\ben}{\begin{equation*}}
\newcommand{\een}{\end{equation*}}
\newcommand{\ban}{\begin{eqnarray*}}
\newcommand{\ean}{\end{eqnarray*}}

\newcommand{\bT}{\mathbf{T}}			
\newcommand{\bP}{\mathrm{\mathbf{P}}}			
\newcommand{\bC}{\mathrm{\mathbf{C}}}			
\newcommand{\bS}{\mathrm{\mathbf{S}}}			
\newcommand{\bH}{\mathrm{\mathbf{H}}}			
\newcommand{\bI}{\mathrm{\mathbf{I}}}			

\newcommand{\Idd}{\mathrm{Id}}			
\newcommand{\MaP}{\left[\bP\right]}		

\newcommand{\tr}{\text{tr}}
\newcommand{\set}[1]{\left\{#1\right\}}                 

\definechangesauthor[name=MO, color=red]{MO}    
\definechangesauthor[name=NA, color=green]{NA}	

\begin{document}

\title{Symmetry classes in piezoelectricity from second-order symmetries}

\author{M. Olive}
\address[Marc Olive]{Université Paris-Saclay, ENS Paris-Saclay, CNRS,  LMT - Laboratoire de Mécanique et Technologie, 91190, Gif-sur-Yvette, France}
\email{marc.olive@math.cnrs.fr}

\author{N. Auffray}
\address[Nicolas Auffray]{Univ Gustave Eiffel, CNRS, MSME UMR 8208, F-77454 Marne-la-Vallée, France}
\email{nicolas.auffray@univ-eiffel.fr}


\date{\today}

\maketitle

\begin{abstract}
The piezoelectricity law is a constitutive model that describes how mechanical and electric fields are coupled within a material. In its linear formulation this law comprises three constitutive tensors of increasing order:  the second order permittivity tensor $\bS$, the third order piezoelectricity tensor $\bP$ and the fourth-order elasticity tensor $\bC$. 
In a first part of the paper, the symmetry classes of the piezoelectricity tensor alone are investigated. Using a new approach based on the use of the so-called \emph{clips operations}, we establish the 16 symmetry classes of this tensor and provide their associated \emph{normal forms}. Second order orthogonal transformations (plane symmetries and $\pi$-angle rotations) are then used to characterize and classify directly 11 out of the 16 symmetry classes of the piezoelectricity tensor. An additional step to distinguish the remaining classes is proposed.
\end{abstract}



\section*{Introduction}

Originally initiated in the field of crystallography, the study of the links between spatial invariances of physical phenomena and the invariances of the underlying matter have now spread all over engineering sciences. The tools initially introduced by Curie~\cite{Cur1894}  are central to contemporary materials science where materials are designed for specific applications \cite{AB03,BKP+09,PMT+19,YCS20}. These materials that can be termed composites, architectured or meta have in common the fact that their internal geometry is specifically designed to produce, or inhibit, physical couplings within the matter. Among the characteristics of an architectured material that can be tailored is the anisotropy of the different physical properties. The anisotropy of a physical property is the angular variation of this property with respect to material directions.  It is important for applications to understand how the physical anisotropy relates to the symmetries of the material. This understanding helps to select the right material for a specific technical application or to design it if it appears that no bulk material with the required properties is available \cite{SFK97,YCS20}.
In continuum physics, physical properties are described by the means of constitutive law which describe how the matter reacts to physical fields  \cite{ZB1994}.  In linear physics,  constitutive laws are modelled using tensors and the questions related to their anisotropy can be formulated in the language of the group representation theory~\cite{FH1991,Ste1994}.  \\

The classical approach to study the anisotropy of a constitutive tensor law is inherited from crystal physics \cite{Nye85}. It consists in studying and classifying the consequences of the invariances imposed by the \red{32} crystallographic point groups \red{(\emph{i.e.} subgroups of the orthogonal group that can preserve a lattice, see~\cite[Appendix A]{Ste1994} for instance)} on the matrix representation of a constitutive tensor. This approach has been automated and there are on-line applications that produce the desired results for any tensor once its order and index symmetries have been specified \cite{GEE+19}.
However, the crystallographic  approach to the  classification of tensorial anisotropies suffers from two important flaws:
\begin{enumerate}
\item  some subgroups of $\OO(3)$ are not crystallographic,  the associated physical anisotropies are absent from the classification\footnote{For instance, the five-fold rotation invariance is not crystallographic since incompatible with the translation lattice, but it can anyway describe the anisotropic properties of some constitutive tensors, see for instance \cite{ADR15}.};
\item different crystallographic groups can lead to apparently different operators which are, in fact, identical modulo an isometric transformation. 
\end{enumerate}
As a result it appears that, for a given tensor, the crystallographic approach can, in the same time, omit some type of anisotropy and give different names to tensors that, in fact, have the same type of anisotropy.  Despite the limitations just mentioned, this approach is very classical and is detailed in many reference monographs \cite{Sch54, Nye85}.  The consequences are confusions which still populate recent and important  publications such as \cite{JCG+15} in which results for a spurious class are provided, while another (crystallographic) class is missing.  All this illustrates the fact that the crystallographic approach is not the right framework for formulating the problem.

An appropriate formalism for dealing with the anisotropy of tensor spaces was set up by Forte--Vianello~\cite{FV1996} in an article devoted to the elasticity tensor. In this reference, the notion of a tensor \emph{symmetry class} is defined for the first time in mechanics. Using a geometric approach,  the authors demonstrated that the 10 anisotropic systems classically considered by crystallographers \cite{Sch54,Roy99} reduces to only 8 distinct symmetry classes, 2 crystallographic systems being fictitious. This approach was a true change of paradigm since, for the first time, the classification problem is formulated with respect to the vector space of constitutive tensors and without referring to any crystallographic system. Since then the Forte-Vianello method has been successfully applied to other constitutive tensor spaces~\cite{FV1997,GW2002,Wel2004,LQH2011}. 
For instance, in the case of the third-order piezoelectricity tensor, the 17 anisotropic systems obtained by the crystallographic approach \cite{Sch54,Roy99}  reduce to 16 symmetry classes \cite{Wel2004}.

In its original setting, the Forte-Vianello approach requires rather fine calculations and reasoning to establish the classification. This complexity make difficult its application to more involved situations, such as constitutive tensors of order greater than 4 or  coupled constitutive laws  involving a family of constitutive tensors. The piezoelectricity  which describes how mechanical and electrical fields interact within the matter is the archetype of such a coupled law. In its linear formulation, this behaviour involves three constitutive tensors of increasing order:  the second-order permittivity tensor $\bS$, the third-order piezoelectricity tensor $\bP$ and the fourth-order elasticity tensor $\bC$ \cite{Lan1984,Roy99}.

The symmetry classes of the piezoelectricity tensor $\bP$ has already been obtained~\cite{GW2002,Zou2013} with strategies that are difficult to generalize.  We propose here an original and general approach using the \emph{clips operations}, as introduced in~\cite{Oli2019}, in order to directly obtain the 16 classes of symmetry sought. In contrast to the Forte-Vianello approach, this method extends directly to more complicated situations. Aside to this classification we also provide:
\begin{enumerate}
\item two different explicit harmonic decompositions~\cite{Bac1970,FV1996} of the piezoelectricity tensor;
\item the normal forms of the piezoelectricity tensor for each of the 16 symmetry classes, given in a Kelvin representation.
\end{enumerate}
Not wanting to limit ourselves to these already known results, we also propose a classification of the piezoelectricity symmetry classes using orthogonal transformations of order $2$, following what has been done in the case of elasticity~\cite{FGB1998,CVC2001}.  In these references which concern the elasticity tensor (even-order), only symmetry planes were considered. In the present contribution, in order to treat odd-order tensors, the method is extended to both \emph{plane symmetries} and \emph{axial symmetries} (i.e rotations of order two). In contrast to the elasticity tensor, for which all symmetry classes can be identified from symmetry planes, 11 of the 16 symmetry classes of the piezoelectric tensor can be directly identified by this approach, and an additional step to distinguish the remaining classes is then proposed.

\subsection*{Organization of the paper} In \autoref{s:Piez}, the mathematical formulation of the physical coupling between electricity and elasticity is recalled. In this section,  the definitions of tensor symmetry groups and classes are introduced and detailed in the context of the piezoelectricity coupling.  In  \autoref{s:Harm_Decomposition}, the harmonic decomposition of the piezoelectricity tensors is introduced. This decomposition, which is the generalization of the decomposition of second order symmetric tensors into spherical and deviatoric parts, is important to understand the physics described by the model. Two particular explicit harmonic decompositions of the space of piezoelectricity tensors are proposed.  
In the following section (\autoref{s:Sym_Class_Normal_Forms}), the existence of the 16 symmetry classes for the piezoelectric tensor is established using clips operations. In \autoref{s:Sym_Classes_from_Second_Order_Sym}, two-order symmetries are used to obtain a direct characterization of 11 out of 16 piezoelectricity symmetry classes. A method to distinguish the remaining classes is proposed. The article is supplemented by three annexes. In~\autoref{sec:O3_closed_subgroups}, the main definitions of the closed $\OO(3)$-subgroups are recalled so to have a better idea of the physical content of our results. Then~\autoref{sec:Norm_Form} is devoted to the explicit matrix expression of normal form in each symmetry class. The \autoref{app:Dim_Fixed_Point} provides some useful formulae about dimension of spaces associated to normal forms. 

\section{Piezoelectricity}\label{s:Piez}

In this section, the  behaviour of a linear piezoelectric solid is detailed. Readers interested in a more complete introduction can refer to~\cite{Sch54,Lan1984,Roy99,Ieee88}.

\subsection{Linear piezoelectricity from electro-mechanical coupling}\label{subsec:Lin_Space_Piez}

In the context of the infinitesimal strain theory, the mechanical state of a material is characterized by two fields of symmetric second-order tensors: the Cauchy stress tensor field $\dTs{\sigma}$ and the infinitesimal strain tensor field $\dTs{\varepsilon}$. These two fields are linked by a constitutive law,  which describes the mechanical behaviour of a specific material over a limited range of external parameters.  For linear elasticity, which is a particular constitutive model, the relation which is  pointwise and linear can be written
\begin{equation*}
\dTs{\sigma}=\dT{C}:\dTs{\varepsilon},\quad \sigma_{ij}=C_{ijlm}\varepsilon_{lm},
\end{equation*}
in which $\dT{C}$ is a fourth-order tensor, known as the \emph{elasticity tensor}~\cite{Gur73,MH1994}.

In the same way, the \emph{electrical} state is described by two vector fields: the electric displacement  $\dT{d}$ and the electric field $\dT{e}$. As in the mechanical situation, these fields are connected by a constitutive law that describes the behaviour of each different material.  For linear conductivity, this relation which is  pointwise and linear can be written
\begin{equation*}
\dT{d}=\dT{S}\cdot\dT{e},\quad d_{i}=S_{im}e_m,
\end{equation*}
in which $\dT{S}$ is a second-order tensor, known as the  \emph{permittivity tensor}~\cite{Lan1984}.
For non-centro symmetric materials these two phenomena are not independent but coupled~\cite{Cur1894,Lan1984}. In this situation the constitutive law reads \cite{Roy99,Wel2004}
\begin{equation}\label{eq:Coupled_Const_Law}
\begin{cases}
\dTs{\sigma}=\dT{C}:\dTs{\varepsilon}-\dT{e}\cdot\dT{P} \\
\dT{d}=\dT{P}:\dTs{\varepsilon}+\dT{S}\cdot\dT{e}
\end{cases}
\end{equation}
in which a third-order tensor $\dT{P}$, known as the \emph{piezoelectricity tensor}, responsible for the coupling appears\footnote{Depending on the considered set of primary variables, four different conventions can be used to express the law of piezoelectricity \cite{Ieee88,Roy99}. The one chosen here is regarded as the most general according to the \emph{IEEE Standard on Piezoelectricity} \cite{Ieee88}. In any case, the results of the present article are essentially independent of the chosen convention (except for the explicit harmonic decompositions).}.
With respect to an orthonormal basis of $\RR^{3}$ the coupled constitutive law can be expressed as follows
\begin{equation}\label{eq:SGE}
\begin{cases}
\sigma _{ij}=C_{ijlm}\varepsilon _{lm}-P_{mij}e_{m} \\ 
d_{i}=P_{ilm}\varepsilon _{lm}+S_{im}e_{m}
\end{cases}.
\end{equation}
At any material point, the linear electromechanical behaviour is defined by a triplet $\bm{\mathcal{P}}$ of constitutive tensors
\begin{equation*}
\bm{\mathcal{P}}:=(\dT{C},\dT{P},\dT{S})\in \Ela\oplus \Pie\oplus \mathbb{S}^2
\end{equation*}
in which:
\begin{enumerate}
\item $\Ela$ is the 21 dimensional vector space of elasticity tensors: 
\begin{equation*}
	\Ela:=\set{\dT{C}\in\otimes^{4}(\RR^3),\quad C_{ijkl}=C_{jikl}=C_{ijlk}=C_{klij}}.
\end{equation*}
\item $\Pie$ is the 18 dimensional vector space of piezoelectricity tensors:
\begin{equation*}
	\Pie:=\set{\dT{P}\in\otimes^{3}(\RR^3),\quad P_{ijk}=P_{ikj}}.
\end{equation*}
\item $\mathbb{S}^2$ is the 6 dimensional vector space of permittivity tensors:
\begin{equation*}
	\mathbb{S}^2:=\set{\mathbf{S}\in\otimes^2(\RR^3),\quad S_{ij}=S_{ji}}.
\end{equation*}
\end{enumerate}

The space of piezoelectricity law will simply be denoted by $\Piez$ in the following, with $\Piez$ being defined as
\ben
\Piez=\Ela\oplus \Pie\oplus \mathbb{S}^2.
\een

\subsection{Symmetry groups and symmetry classes}

Consider a homogeneous piezoelectric material oriented in some way with respect to a given and fixed reference, for instance a testing device. With respect to this reference, its electromechanical behaviour is characterized by a triplet $\bm{\mathcal{P}}$ of tensors
\begin{equation*}
\bm{\mathcal{P}}:=(\dT{C},\dT{P},\dT{S})\in\Piez.
\end{equation*}
Consider an orthogonal transformation $g\in \OO(3)$ acting on the material. The electromechanical behaviour of the material in its new configuration is characterized by another triplet $\overline{\bm{\mathcal{P}}}:=(\overline{\dT{C}},\overline{\dT{P}},\overline{\dT{S}})$ defined in coordinates by:
\begin{align}\label{eq:New_Tensors}
	\overline{C}_{ijkl}&:=g_{ip}g_{jq}g_{kr}g_{ls}C_{pqrs},\quad \overline{P}_{ijk}:=g_{ip}g_{jq}g_{kr}P_{pqr},\quad
	\overline{S}_{ij}:=g_{ip}g_{jq}S_{pq}.
\end{align}
Obviously, the physical nature of the piezoelectric material is not affected by this transformation, only its constitutive tensors are transformed.
Each of these transformations corresponds to a specific case of a \emph{linear representation} of the group $\OO(3)$~\cite{IG1984,Ste1994}, so that the equations \eqref{eq:New_Tensors} can  be recast as:
\begin{equation}\label{eq:Rep_Tensorielles}
	\overline{\dT{C}}=g\star \dT{C},\quad \overline{\dT{P}}=g\star \dT{P},\quad \overline{\dT{S}}=g\star \dT{S},
\end{equation}
in which the \emph{standard} linear representation of $\OO(3)$ on the space $\otimes^{n}(\RR^{3})$ of $n$th-order tensors is given in any orthonormal basis by
\begin{equation}\label{eq:standard_rep}
	(g\star \bT)_{i_1\dotsc i_n}:=g_{i_1j_1}\dotsc g_{i_nj_n}T_{j_1\dotsc j_n},\quad g\in \OO(3),\quad \bT\in \otimes^{n}(\RR^{3}).
\end{equation}
\begin{rem}
There is another linear representation of $\OO(3)$ on the space of $n$th order tensors, called the \emph{twisted} representation
\begin{equation}\label{eq:twist_rep}
g\: \hat{\star}\: \bT:=\det(g) (g\star \bT)
\end{equation}
where such tensor $\bT$ is said to be a \emph{pseudo--tensor}. Note that such a twisted representation is used in the following section (\autoref{s:Harm_Decomposition}).
\end{rem}

The possible \emph{anisotropies} of a constitutive law are modelled on the \emph{symmetry classes} of the associated representation. First of all, for a given constitutive tensor $\bT$, its \emph{symmetry group} is defined as the set of orthogonal transformations letting $\bT$ invariant:
\begin{equation*}
G_{\bT}:=\set{g\in \OO(3),\quad g \star \bT=\bT}.
\end{equation*}
For any constitutive tensor $\overline{\bT}=g\star \bT$ following a new orientation given by $g\in \OO(3)$, the symmetry group $G_{\overline{\bT}}$ of $\overline{\bT}$ is related to $G_{\bT}$ by
\begin{equation*}
G_{\overline{\bT}}=gG_{\bT}g^{-1}.
\end{equation*}
Thus, the type of anisotropy of a tensor is not described by its symmetry group, which refers to a specific orientation, but rather by its symmetry class
\begin{equation*}
[G_{\bT}]:=\set{gG_{\bT}g^{-1},\quad g\in \OO(3)}
\end{equation*}
which is  the \emph{conjugacy class} of its symmetry group (the set of all subgroups conjugate to $G_{\bT}$). 

In the following, attention will be restricted from the complete constitutive law to only its coupling component encoded by $\dT{P}\in\Pie$. Equivalent studies and results for $\Ela$ and $\mathbb{S}^2$ can be found in other references (see~\cite{FV1996,AKP2014} for instance). From many points of view, knowing the number and the type of symmetry classes of given tensor space is an important questions~\cite{FV1996,FV1997,LQH2011,OA2014,OA2014a}. 

In the next two sections, the basic tools needed to decide these issues for $\Pie$ will be introduced and applied: the first tool is the \emph{harmonic decomposition} and the second one are the \emph{clips operations}.

\section{Harmonic decomposition of the space of piezoelectricity tensors}\label{s:Harm_Decomposition}

As a very classical result, any second-order symmetric tensor can be decomposed into a deviatoric part (which is symmetric and traceless) and a spherical part. The generalization of this decomposition to tensor spaces of any order is known as the \emph{harmonic decomposition}~\cite{Bac1970,FV1996}, where we need to introduce the spaces of higher-order deviators, that are the spaces $\HH^n$ of $n$th-order harmonic tensors (definition~\ref{def:Harm_Tensors}). 
To define these spaces,  let us first introduce the space $\Sym^{n}$ of $n$th-order totally symmetric tensors on $\RR^{3}$. Take now $\tr(\bS)$ to be the totally symmetric tensor of order $n-2$ in which the contraction on the first two indices of $\bS$ has been achieved\footnote{Since the contraction is on a fully symmetric tensor, it does not matter which indices are contracted. First indices are just considered here for  simplicity.}. 
\begin{defn}\label{def:Harm_Tensors}
	An $n$th-order totally symmetric and \emph{traceless} tensor will be called a \emph{harmonic tensor} and the subspace of $\Sym^{n}$ of harmonic tensors will be denoted by $\HH^{n}$. On  $\RR^3$, it is a $2n+1$ dimensional vector space. 
\end{defn}
Each space of harmonic tensors inherits the standard and the twisted $\OO(3)$-representations, respectively given by \eqref{eq:standard_rep} and~\eqref{eq:twist_rep}. To make a clear distinction between these two representations, the notation $\HH^{n}$ refers to space of $n$th-order harmonic tensors endowed with the standard representation, while $\HH^{\sharp n}$ refers to the same space endowed with the twisted representation. Each of these representations are \emph{irreducible}, meaning that there is no non--trivial stable subspace~\cite{IG1984}.
The harmonic structure of $\Pie$ can easily be revealed by noting that 
\begin{equation*}
	\Pie\simeq \HH^{1}\otimes(\HH^2\oplus\HH^0)
\end{equation*}
and using Clebsch-Gordan product for decomposing tensorial product of irreducible spaces into a sum of irreducible ones \cite{JCB1978}. It results that
\begin{equation*}
\Pie\simeq \HH^{3}\oplus \HH^{\sharp 2}\oplus \HH^{1}\oplus \HH^{1}.
\end{equation*}

Let consider now explicit harmonic decompositions of $\Pie$, meaning explicit isomorphisms from $\Pie$ to $\HH^{3}\oplus \HH^{\sharp 2}\oplus \HH^{1}\oplus \HH^{1}$. Such explicit isomorphisms are not unique and different decomposition can be found in the literature \cite{Spe70,Wel2004,ZZD+01,Zou2013}. Here two explicit decompositions will be provided
\begin{enumerate}
\item The \emph{Clebsch-Gordan} harmonic decomposition \cite{AAD20}, which consists in decomposing the piezoelectricity tensor in two parts: one generating an electric displacement from a deviatoric strain, the other one generating an electric displacement from a hydrostatic strain. These blocks are then decomposed into irreducible harmonic tensors. The  specific expressions are given here by~\eqref{eq:Cov_Lienaire_Piezo} and~\eqref{eq:Cov_Harm3_Piezo}. This decomposition which possesses a  clear physical content is interesting for physical applications and corresponds, up to some scaling factors,  to the one proposed in \cite{Zou2013}. 
\item The \emph{Schur-Weyl}  harmonic decomposition, consists in first decomposing the piezoelectricity tensor according to its index symmetries before proceeding to the harmonic decomposition of each part. This decomposition follow the lines of the method used by Backus in the case of the elasticity tensor \cite{Bac1970}. This decomposition will find interesting application in \autoref{s:Sym_Classes_from_Second_Order_Sym} for symmetry classes identification.
\end{enumerate}

\begin{rem}
Since the third-order piezoelectricity tensor has a similar structure to the third-order strain-gradient elasticity tensor that appears in Mindlin Strain Gradient Elasticity model~\cite{Min1964}, numerous results concerning its harmonic decomposition can be found in this literature, for instance \cite{HF97,LYC+03,Auf15,Laz16}.
It has to be noted that in this literature the Schur-Weyl decomposition, which corresponds to the Mindlin type III formulation of Strain Gradient Elasticity, is the most natural.
\end{rem}

\subsection{Clebsch--Gordan decomposition of piezoelectricity tensors}
Let  us consider some piezoelectricity tensor $\bP\in \Pie$ and let us first define $(\bH,\ba,\uu,\vv)\in \HH^{3}\oplus \HH^{\sharp 2}\oplus \HH^{1}\oplus \HH^{1}$ by:
\begin{equation}\label{eq:Cov_Lienaire_Piezo}
a_{ij}:=\frac{1}{2}\left(\epsilon_{ipq}P_{pqj}+\epsilon_{jpq}P_{pqi}\right),\quad u_i:=P_{ppi}-\dfrac{1}{3}P_{ipp},\quad 
v_i:=P_{ipp}
\end{equation}
and 
\begin{equation}\label{eq:Cov_Harm3_Piezo}
H_{ijk}:=P_{ijk}-\frac{1}{3}\left(\epsilon_{ijp}a_{pk}+\epsilon_{ikp}a_{pj}\right)-\frac{3}{10}\left(u_{j}\delta_{ik}+u_{k}\delta_{ij}-\frac{2}{3}u_{i}\delta_{jk} \right)-\frac{1}{3}v_{i}\delta_{jk}.
\end{equation}
From this, the piezoelectricity tensor $\bP$ can be decomposed into 
\begin{equation}\label{eq:Piez_Dev_Sph}
	\bP=\bP^{(1,2)}+\bP^{(1,0)}
\end{equation}
with
\begin{eqnarray*}
	P^{(1,2)}_{ijk}&:=&H_{ijk}+\frac{1}{3}\left(\epsilon_{ijp}a_{pk}+\epsilon_{ikp}a_{pj}\right)+\frac{3}{10}\left(u_{j}\delta_{ik}+u_{k}\delta_{ij}-\frac{2}{3}u_{i}\delta_{jk} \right),\\
	P^{(1,0)}_{ijk}&:=&\frac{1}{3}v_{i}\delta_{jk},
\end{eqnarray*}
where the notation $\bP^{(p,q)}$ indicates a linear map from $\HH^{q}$ to $\HH^{p}$.
Finally, we claim that the map
\begin{equation}
	\phi\: : \: \bP\in \Pie \mapsto (\bH,\ba,\uu,\vv)\in \HH^{3}\oplus \HH^{\sharp 2}\oplus \HH^{1}\oplus \HH^{1}
\end{equation}
is an \emph{equivariant isomorphism}, where equivariance means that
\begin{equation*}
\forall g\in\OO(3),\quad	\phi(g\star \bP)=g\star(\bH,\ba,\uu,\vv)=(g\star \bH,g \: \hat{\star} \: \ba,g\star \uu,g\star \vv).
\end{equation*}
Note that an explicit inverse linear map of $\phi$ is given by 
\begin{multline*}
	(\bH,\ba,\uu,\vv)\in \HH^{3}\oplus \HH^{\sharp 2}\oplus \HH^{1}\oplus \HH^{1}\mapsto \bP\in \Pie,\\
			P_{ijk}:=H_{ijk}+\frac{1}{3}\left(\epsilon_{ijp}a_{pk}+\epsilon_{ikp}a_{pj}\right)+\frac{3}{10}\left(u_{j}\delta_{ik}+u_{j}\delta_{ik}-\frac{2}{3}u_{i}\delta_{jk} \right)+\frac{1}{3}v_{i}\delta_{jk}.
\end{multline*}
The decomposition~\eqref{eq:Piez_Dev_Sph} has a clear physical content. Indeed, the tensor $\bP^{(1,2)}$ is the part of $\bP$ which generates an electric displacement from a deviatoric strain $\dTs{\varepsilon}^{d}$, while $\bP^{(1,0)}$ is the electric displacement generated from a hydrostatic strain $\dTs{\varepsilon}^{h}$ with
\begin{equation}
	\dTs{\varepsilon}^{h}=\frac{1}{3}\left(\dTs{\varepsilon}:\bI\right)\bI\quad;\quad \dTs{\varepsilon}^{d}=\dTs{\varepsilon}-\dTs{\varepsilon}^{h},
\end{equation}
and $\bI$ is the second order identity tensor. 

This way to decompose the piezoelectricity tensor is interesting for optimal design applications \cite{NLZ+16}. For instance, an effective material for which $\bP^{(1,2)}=0$, is a material for which the electromechanical coupling only show off in case of hydrostatic strain.  

\subsection{Schur--Weyl decomposition of piezoelectricity tensors}
Let us take back here piezoelectricity tensor $\bP\in \Pie$ and define $\ba\in \HH^{\sharp 2}$ using the same formula~\eqref{eq:Cov_Lienaire_Piezo}. We then consider
\begin{equation}\label{eq:Vect_Schur_Weyl}
u'_i:=\frac{1}{3}\left(P_{ipp}+2P_{ppi}\right),\quad v'_i:=\frac{1}{3}\left(P_{ppi}-P_{ipp}\right)
\end{equation}
and
\begin{equation*}
H_{ijk}:=\frac{1}{3}\left(P_{ijk}+P_{kij}+P_{jik}\right)-
\frac{1}{5}\left(\delta_{ij}u'_k+\delta_{ik}u'_j+\delta_{jk}u'_i\right)
\end{equation*}
where such third order harmonic tensor is exactly the same as the one defined by~\eqref{eq:Cov_Harm3_Piezo}.
From this, piezoelectricity tensor $\bP$ can decomposes into 
\begin{equation}\label{eq:Piez_GL}
	\bP=\bP^{s}+\bP^{r}
\end{equation}
with
\begin{eqnarray*}
	P^{s}_{ijk}&:=&H_{ijk}+\frac{1}{5}\left(\delta_{ij}u'_k+\delta_{ik}u'_j+\delta_{jk}u'_i\right),\\
	P^{r}_{ijk}&:=&\frac{1}{3}\left(\epsilon_{ijp}a_{pk}+\epsilon_{ikp}a_{pj}\right)-\frac{1}{2}\left(2\delta_{jk}v'_i-\delta_{ij}v'_k-\delta_{ik}v'_j\right).
\end{eqnarray*}
where $\bP^{s}$ is a completely symmetric tensor and $\bP^{r}$ is a remainder. Now the linear map
\begin{equation*}
\bP\in \Pie\mapsto (\bH,\ba,\uu',\vv')\in \HH^{3}\oplus \HH^{\sharp 2}\oplus \HH^{1}\oplus \HH^{1}
\end{equation*}
is an equivariant isomorphism, with inverse map given by
\begin{multline}
(\bH,\ba,\uu',\vv')\in \HH^{3}\oplus \HH^{\sharp 2}\oplus \HH^{1}\oplus \HH^{1}\mapsto \bP\in \Pie,\\
P_{ijk}=H_{ijk}+\frac{1}{3}\left(\epsilon_{ijp}a_{pk}+\epsilon_{ikp}a_{pj}\right)
+\frac{1}{5}\left(\delta_{ij}u'_k+\delta_{ik}u'_j+\delta_{jk}u'_i\right)-\frac{1}{2}\left(2\delta_{jk}v'_i-\delta_{ij}v'_k-\delta_{ik}v'_j\right).
\end{multline}
This decomposition will find interesting application in \autoref{s:Sym_Classes_from_Second_Order_Sym} for symmetry class identification.
\begin{rem}\label{rem:Different_Piez_Harm_Decomp}
	We gave here two explicit isomorphisms
	\begin{equation*}
	\Pie\simeq \HH^3\oplus\HH^{\sharp 2}\oplus \HH^{1}\oplus \HH^{1}
	\end{equation*}
	and we know that any linear isomorphism
	\begin{equation*}
	(\uu,\vv)\in \HH^{1}\oplus \HH^{1} \mapsto (\tilde{\uu},\tilde{\vv})\in \HH^{1}\oplus \HH^{1}
	\end{equation*}
	can lead to another harmonic decomposition $(\bH,\ba,\tilde{\uu},\tilde{\vv})$. Other examples of such harmonic decomposition can be found in~\cite{Wel2004,Auf2017}.
\end{rem}

\section{Symmetry classes of piezoelectricity tensor space}\label{s:Sym_Class_Normal_Forms}

The geometric definition of a \emph{symmetry class} was first given by Forte--Vianello~\cite{FV1996} in the particular case of the elasticity tensor. The authors have introduced  an original approach to obtain the 8 symmetry classes of $\Ela$. Since then the method has been applied to other constitutive tensors: piezoelectricity~\cite{GW2002,Wel2004}, photoelasticity~\cite{FV1997} and flexoelectricity~\cite{LQH2011}. Since the original method was a bit tedious to apply,  a general algorithm was  proposed by the present authors in \cite{OA2013,Oli2014a,OA2014,Oli2019} to simplify the determination of symmetry classes. This method is based on the definition of \emph{clips operations} on conjugacy classes, a strategy which was initiated in the 90' by Chossat \emph{et al.}~\cite{CLM1990,CG1994,CG1996}.
We propose to get back to the determination of the symmetry classes of the piezoelectricity tensor using this approach. As will be seen, by this means,  the 16 symmetry classes of $\Pie$ are directly obtained.  For clarity of presentation, the notations and definitions of $\OO(3)$-subgroups have been moved to \autoref{sec:O3_closed_subgroups}. The classification theorem is given by:

\begin{thm}\label{thm:Sym_Class_Piezo}
\red{A conjugacy class of a closed subgroup of $\OO(3)$ is a symmetry class of a tensor $\bP\in \Pie$ if and only if it belongs to the 16 elements set of symmetry classes}:
\begin{equation*}
\{[\triv],[\ZZ_2],[\ZZ_{3}],[\DD^{v}_2],[\DD^{v}_{3}],[\ZZ^{-}_2],[\ZZ^{-}_{4}],[\DD_2],[\DD_{3}],[\DD^{h}_4],[\DD^{h}_{6}],[\SO(2)],[\OO(2)],[\OO(2)^{-}],[\octa^-],[\OO(3)]\}.
\end{equation*}
\end{thm}

The \emph{normal forms} associated to the non-trivial symmetry classes are provided in \autoref{sec:Norm_Form}.

\begin{rem} 
It can be observed that the following two classes are absent from the set of  symmetry classes of $\Pie$:
\begin{itemize}
\item The class $[\SO(3)]$ is missing. This is due to the fact that there is no any pseudo-scalar in the harmonic structure of $\mathbb{P}\mathrm{iez}$, i.e. element of $\HH^{\sharp 0}$. The lacking of this class is somehow exceptional since for a generic odd order tensor space this class is present \cite{OA2014}.
\item The class $[\ZZ^{-}_{6}]$ is  missing. Note that a $\ZZ_6^-$-invariant piezoelectric tensor $\bP$ can be constructed. But some $g\in \OO(3)$ can then be found such that the symmetry class of $g\star \bP$ is exactly $[\DD^{h}_{6}]$ (see Remark~\ref{rem:Z6-form}). The absence of the class $[\ZZ^{-}_{2m}]$ with $m=2p+1$ for a $m$th-order tensor is generic \cite{OA2014}.
This result can be compared with the piezoelectricity anisotropic systems detailed for crystallographic groups, for instance in Schouten \cite{Sch54}. In these tables 17 different "matrix forms" are proposed, the extra matrix is associated to a material having $\ZZ^{-}_{6}$ for point group.
\end{itemize}
\end{rem}

The set of conjugacy classes of a compact group is a \emph{partially ordered set} (\emph{poset})~\cite{Bre1972}. For two given conjugacy classes $[H_1]$ and $[H_2]$, we write $[H_1]\preceq [H_2]$ iff $H_1\subset gH_2g^{-1}$ for some $g\in \OO(3)$. As a result, relations between the symmetry classes of piezoelectricity tensors can be graphically organized as a lattice, shown in \autoref{fig:Lattice-piezo}, in which an arrow between two classes $[H_1]\rightarrow\ [H_2]$ indicates that $[H_1]\preceq [H_2]$.
\begin{figure}[H]
  \centering
  \includegraphics[scale=0.8]{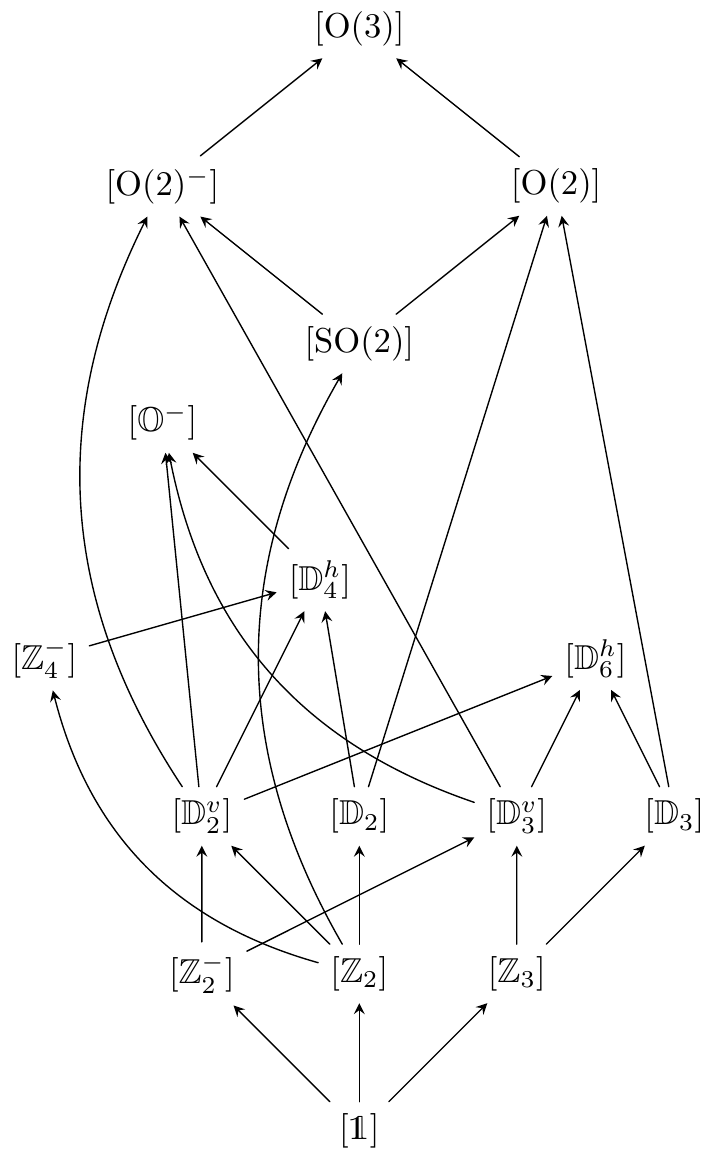}
  \caption{Lattice of symmetry classes of $\Pie$.}
  \label{fig:Lattice-piezo}
\end{figure}

To obtain the symmetry classes of $\Pie$, which is a  $\OO(3)$-representation,  \emph{clips operations} will be used \cite{OA2013,OA2014,Oli2019}. All details can be found in~\cite{OA2013,OA2014,Oli2014a} but we summarize here the main ideas.
Clips operations on conjugacy classes $[H_i]$ of closed $\OO(3)$-subgroups are defined as follows:
\begin{equation*}
	[H_1]\circledcirc[H_2]:=\bigcup_{g\in \OO(3)}H_1\cap\left(gH_2g^{-1}\right). 
\end{equation*}
For two given $\OO(3)$-representations $V_1$ and $V_2$ with sets of symmetry classes $\mathfrak{I}(V_1)$ and $\mathfrak{I}(V_2)$, the set of symmetry classes $\mathfrak{I}(V_1\oplus V_2)$ is given by 
\begin{equation*}
	\mathfrak{I}(V_1\oplus V_2)=\mathfrak{I}(V_1)\circledcirc\mathfrak{I}(V_2):=\bigcup_{[H_i]\in \mathfrak{I}(V_i)}\left( [H_1]\circledcirc[H_2]\right). 
\end{equation*}

\begin{proof}[Proof of Theorem~\ref{thm:Sym_Class_Piezo}]
We first recall the harmonic decomposition
\begin{equation*}
	\Pie\simeq \HH^3\oplus \HH^{\sharp 2}\oplus \HH^1\oplus \HH^1
\end{equation*}
and results from~\cite{IG1984} to obtain symmetry classes of each irreducible representation that appear in $\Pie$
\begin{align*}
\mathfrak{I}(\HH^1)&=\set{[\OO(2)^-],[\OO(3)]},\quad \mathfrak{I}(\HH^{\sharp 2})=\set{[\DD_2],[\DD_4^h],[\OO(2)],[\OO(3)]},\\
\mathfrak{I}(\HH^3)&=\set{[\triv],[\ZZ_2^-],[\DD_2^v],[\DD_3^v],[\DD_6^h],[\octa^-],[\OO(2)^-],[\OO(3)]}.
\end{align*}
To obtain $\mathfrak{I}(\HH^3\oplus\HH^{\sharp 2})$ we use the clips table \ref{tab:Clips_Table_H3H2*} (taken from~\cite{Oli2019}).
\begin{center}
\begin{table}[h!]
\resizebox{\textwidth}{!}{%
\renewcommand{\arraystretch}{1.2}
\begin{tabular}{|c|c|c|c|c|c|c|c|c|}
\hline
$\circledcirc$ & $[\triv]$ & $[\ZZ_2^-]$ & $[\DD_2^v]$ & $[\DD_3^v]$ & $[\DD_6^h]$ & $[\octa^-]$ & $[\OO(2)^-]$ & $[\OO(3)]$ \\
\hline \hline
$[\DD_2]$ & $[\triv]$ & $[\triv]$ & $[\triv]$ & $[\triv]$ & $[\triv]$ & $[\triv]$ & $[\triv]$ & $[\DD_2]$  \\ \hline
$[\DD_4^h]$ & $[\triv]$ & $[\triv],[\ZZ_2^-]$ & \begin{tabular}{c} $[\triv],[\ZZ_2],$ \\ $[\DD_2^v]$ \end{tabular} & $[\triv],[\ZZ_2^-]$ &
\begin{tabular}{c} $[\triv],[\ZZ_2],$ \\ $[\ZZ_2^-],[\DD_2^v]$ \end{tabular}  &
\begin{tabular}{c} $[\triv],[\ZZ_4^-],$ \\ $[\DD_4^h]$ \end{tabular} & \begin{tabular}{c} $[\triv],[\ZZ_2],$ \\ $[\DD_2^v]$ \end{tabular} 
& $[\DD_4^h]$ \\ \hline
$[\OO(2)]$ & $[\triv]$  & $[\triv]$ & $[\triv],[\ZZ_2]$ & $[\triv],[\ZZ_3]$ & 
\begin{tabular}{c} $[\triv],[\DD_2],$ \\ $[\DD_3]$ \end{tabular} 
&
\begin{tabular}{c} $[\triv],[\ZZ_2],$ \\ $[\DD_2],[\DD_3]$ \end{tabular} & 
\begin{tabular}{c} $[\triv],[\ZZ_2],$ \\ $[\SO(2)]$ \end{tabular} & $[\OO(3)]$ \\ \hline
$[\OO(3)]$ & $[\triv]$ & $[\ZZ_2^-]$ & $[\DD_2^v]$ & $[\DD_3^v]$ & $[\DD_6^h]$ & $[\octa^-]$ & $[\OO(2)^-]$ & $[\OO(3)]$\\ \hline
\end{tabular}
}
\caption{Clips table $\mathfrak{I}(\HH^3)\circledcirc\mathfrak{I}(\HH^{\sharp 2})$}
\label{tab:Clips_Table_H3H2*}
\end{table}
\end{center}
Thus we have
\begin{equation*}
\begin{split}
	\mathfrak{I}(\HH^3\oplus\HH^{\sharp 2})=\left\{[\triv],[\ZZ_2],[\ZZ_2^-],[\DD_2],[\DD_2^v],[\ZZ_3],[\DD_3],[\DD_3^v],[\ZZ_4^-],[\DD_4^h],
	[\DD_6^h],\right.\\ 
	\left. [\octa^-],[\SO(2)],[\OO(2)],[\OO(2)^-],[\OO(3)]\right\}.
\end{split}
\end{equation*}
As $[\OO(2)^-]\circledcirc[\OO(2)^-]=\set{[\ZZ_2^-],[\OO(2)^-]}$ we obtain
\begin{equation*}
	\mathfrak{I}(\HH^1\oplus\HH^1)=\set{[\ZZ_2^-],[\OO(2)^-],[\OO(3)]}.
\end{equation*}
Finally, using tables from~\cite{Oli2019} it can be checked by direct inspection that 
\ben
\mathfrak{I}(\Pie)=\mathfrak{I}(\HH^3\oplus\HH^{\sharp 2})\circledcirc \mathfrak{I}(\HH^1\oplus\HH^1)=\mathfrak{I}(\HH^3\oplus\HH^{\sharp 2})
\een
which leads to the conclusion. 
\end{proof}

\section{Symmetry classes characterisation using order two symmetries}\label{s:Sym_Classes_from_Second_Order_Sym}

Following the characterization of the elasticity symmetry classes using plane symmetries as introduced in~\cite{CVC2001}, the purpose of this section is to characterize as many piezoelectric symmetry classes as possible using only order two symmetries, which are elements $g\in\OO(3)$ satisfying $g^2=\Idd$ \red{and $g\neq \Idd$}, where $\Idd$ is the identity element of $\OO(3)$. Such an element is either:
\begin{itemize}
\item a \emph{plane symmetry} $\bs(\bnn):=\mathbf{I}-2\bnn\otimes \bnn$, with $\bnn$ being a unit vector in $\RR^{3}$ \red{and $\mathbf{I}$ the unit second order tensor};
\item an \emph{axial symmetry} $\vR(\bnn,\pi)=-\bs(\bnn)$.
\end{itemize}
The notation $\vR(\bnn,\theta)$, with $\theta\in [0;2\pi[$ denotes a rotation of angle $\theta$ and axis generated by $\bnn$, with the convention that $\vR(\ve_3,\theta)$ has the following matrix representation in the canonical basis $(\ve_1,\ve_2,\ve_3)$ of $\RR^3$
\begin{equation*}
\vR(\ve_3,\theta)=
\begin{pmatrix}
\cos(\theta) & - \sin(\theta) & 0 \\
\sin(\theta) & \cos(\theta) & 0 \\
0 & 0 & 1  	
\end{pmatrix}.
\end{equation*}

As illustrated by~\autoref{fig:Piez_Sym_Class_Second_Order} and~\autoref{fig:Piez_Sym_Class_Second_Order_II}, it is possible to characterize unambiguously $11$ over the $16$ piezoelectricity symmetry classes using only order two symmetries. The remaining classes, which are divided in two sets, can be distinguished by studying invariance with respect to higher-order rotations.

Before detailing the classification obtained, let us explain how  \red{two order symmetries}  are explicitly tested.  Using results from~\cite{ODKD2020a,ODKD2020b}, we provide in Theorems~\ref{thm:Piezo_Plane_sym}--\ref{thm:Piezo_Axis_sym} the explicit equations to determine axes and/or directions related to such order two symmetries.
\subsection{Explicit equations for order two symmetries}
For any given piezoelectricity tensor $\bP\in \Pie$, an order two symmetry is define to be some $g=\vR(\bnn,\pi)$ or $g=\bs(\bnn)$ such that 
\begin{equation}\label{eq:Second_order_sym}
	g\star \bP=\bP.
\end{equation}
Thus a direct way to obtain all order two symmetries of some piezoelectricity tensor $\bP\in \Pie$ is to try to solve
\begin{equation*}
	g\star\bP-\bP=0,\quad g=\bs(\bnn) \text{ or } g=\mathbf{r}(\bnn,\pi)=-\bs(\bnn),
\end{equation*}
with $\bnn=(n_1,n_2,n_3)$ a unit vector in $\RR^{3}$ as unknown, leading to 18 algebraic equations of homogeneous degree $6$ in $(n_1,n_2,n_3)$. In fact, following~\cite{ODKD2020a,ODKD2020b}, it is possible to obtain reduced algebraic equations of degree 3. 
To formulate such reduced equations we need to introduce the totally symmetric part $\bP^s$ of $\bP$, given by
\begin{equation*}
	(\bP^s)_{ijk}:=\frac{1}{3}\left( P_{ijk}+P_{jik}+P_{kij}\right).
\end{equation*}
The algebraic equations characterizing invariance of $\bP$  with respect to a plane symmetry are then obtained~\cite[Theoreom 6.1 (1)]{ODKD2020b}:
\begin{thm}\label{thm:Piezo_Plane_sym}
	 Let $\bP\in \Pie$ be a piezoelectricity tensor, $\bP^{s}$ its totally symmetric part, $\vv'\in \RR^3$ and $\ba\in \Sym^{2}$ respectively defined by~\eqref{eq:Vect_Schur_Weyl} and~\eqref{eq:Cov_Lienaire_Piezo}.  
	 For any unit vector $\bnn=(n_1,n_2,n_3)\in \RR^{3}$, $\bs(\bnn)$ is a plane symmetry of $\bP$ if and only if
	 \begin{equation}\label{eq:Planar_Piezo}
	 \begin{cases}
	 \bP^s \cdot \bnn-[\bnn\otimes \left(\bnn\cdot \bP^s\cdot \bnn\right)+\left(\bnn\cdot \bP^s\cdot \bnn\right)\otimes \bnn]=0, \\
	 \ba-[\bnn\otimes(\ba\cdot \bnn)+(\ba\cdot \bnn)\otimes \bnn]=0 \\
	 \vv'\cdot \bnn=0
	 \end{cases}
	 \end{equation}
	 where $\cdot$ is the standard contraction operation. 
\end{thm}

\begin{rem}
	In any orthonormal basis of $\RR^{3}$ equations~\eqref{eq:Planar_Piezo} reads
	\begin{equation*}
	\begin{cases}
	(P^s)_{ijp}n_p-(P^s)_{piq}n_jn_pn_q-(P^s)_{pjq}n_in_pn_q=0 \\
	a_{ij}-a_{ip}n_jn_p-a_{jp}n_in_p=0 \\
	v'_in_i=0
	\end{cases}
	\end{equation*}
\end{rem}
The algebraic equations related to axis symmetry are then obtained using the \emph{generalized cross product} (see~\cite{DADKO2019}). In the specific case of a totally symmetric tensor $\mathbf{S}\in \Sym^3$ and a first-order tensor $\uu\in\RR^3$, this generalized cross product writes:
\begin{equation}
	(\mathbf{S}\times\uu)_{ijk}:=-\frac{1}{3}\left(S_{ijp}\epsilon_{pkq}u_{q}+S_{ikp}\epsilon_{pjq}u_{q}+S_{kjp}\epsilon_{piq}u_{q}\right)	
\end{equation}
while for two first-order tensors $\uu_1$ and $\uu_2$, $\uu_1\times \uu_2$ is the standard cross product.  We also need to introduce the totally symmetric tensor product between a first order tensor $\uu$ and a second order symmetric tensor $\ba$
\begin{equation}
(\uu\odot \mathbf{a})_{ijk}:=\frac{1}{3}\left(u_{i}a_{jk}+u_{k}a_{ij}+u_{j}a_{ki}\right).	
\end{equation}
Finally, from~\cite[Theoreom 6.1 (2)]{ODKD2020b}, the algebraic equations characterizing the invariance of $\bP$  with respect to axis symmetry are
\begin{thm}\label{thm:Piezo_Axis_sym}
	Let $\bP\in \Pie$ be a piezoelectricity tensor, $\bP^{s}$ its totally symmetric part, $\vv'\in \RR^3$ and $\ba\in \Sym^{2}$ respectively defined by~\eqref{eq:Vect_Schur_Weyl} and~\eqref{eq:Cov_Lienaire_Piezo}. For any unit vector $\bnn=(n_1,n_2,n_3)\in \RR^{3}$, $\vR(\bnn,\pi)$ is an axis symmetry of $\bP$ if and only if
	 \begin{equation}\label{eq:Axis_Piezo}
	 \begin{cases}
	 \left[\bP^s-3\bnn\odot(\bP^s\cdot \bnn)\right]\times \bnn =0 \\
	 \left(\ba\cdot\bnn\right)\times \bnn =0 \\
	 \vv'\times  \bnn=0
	 \end{cases}.
	 \end{equation}
\end{thm}
\begin{rem}
For clarity sake,  the tensor $\left[\bP^s-3\bnn\odot(\bP^s\cdot \bnn)\right]$ can be expressed as follows in any orthonormal basis of $\RR^{3}$ 
\ben
\left[\bP^s-3\bnn\odot(\bP^s\cdot \bnn)\right]_{ijk}=P^s_{ijk}-\left(n_{i}P^s_{jkp}n_{p}+n_{j}P^s_{ikp}n_{p}+n_{k}P^s_{ijp}n_{p}\right).
\een
\end{rem}

\subsection{Order two symmetries and symmetry classes}

To begin, we  propose in \autoref{tab:Second_Order_Sym_groups} to make explicit all order two symmetries contained in a subgroup $H$ taken from one of the $16$ piezoelectricity symmetry classes $[H]$ (see Theorem~\ref{thm:Sym_Class_Piezo}).

\begin{center}
	\small
	\begin{table}[h!]
		\begin{tabular}{|l|l|}
			\hline
			\begin{tabular}{l}
				(1) $\triv,\ZZ_{3}$:\\
				$\bullet$ no any order two symmetries. 
			\end{tabular}
			&
			\begin{tabular}{l}
				(2) $\ZZ_{2}, \ZZ_{4}^{-},\SO(2)$:\\
				$\bullet$  1 axial symmetry $\vR(\ve_{3},\pi)$. 
			\end{tabular}
			\\ \hline
			\begin{tabular}{l}
				(3) $\ZZ_{2}^{-}$:\\
				$\bullet$ 1 plane symmetry $\bs(\ve_{1})$. 
			\end{tabular}
			&
			\begin{tabular}{l}
				(4) $\DD_{2}$:\\
				$\bullet$  3 axial symmetries $\vR(\ve_{i},\pi),i=1,2,3$. 
			\end{tabular}
			\\ \hline
			\begin{tabular}{l}
				(5) $\DD_{2}^{v}$:\\
				$\bullet$ 1 axial symmetry $\vR(\ve_{3},\pi)$;\\
				$\bullet$ 2 plane symmetries $\bs(\ve_{1}),\bs(\ve_{2})$.
			\end{tabular}
			&
			\begin{tabular}{l}
				(6) $\DD_{3}$:\\
				$\bullet$ 3 axial symmetries;\\
				$\vR(\ve_{1},\pi),\vR(\bnn_{k},\pi), \bnn_{k}=\vR\left(\ve_{3},\frac{2k\pi}{3}\right)\ve_{1}, k=1,2$
			\end{tabular}
			\\ \hline
			\begin{tabular}{l}
				(7) $\DD_{3}^{v}$:\\
				$\bullet$ 3 plane symmetries $\bs(\ve_{1}),\bs(\bnn_{k})$\\
				with $\bnn_{k}$ as in $\DD_{3}$ case. 
			\end{tabular}
			&
			\begin{tabular}{l}
				(8) $\DD_{4}^{h}$:\\
				$\bullet$ 3 axial symmetries $\vR(\ve_{i},\pi),i=1,2,3$;\\
				$\bullet$ 2 plane symmetries $\bs(\ve_{2}\pm\ve_{1})$
			\end{tabular}
			\\ \hline
			\begin{tabular}{l}
				(9) $\DD_{6}^{h}$:\\
				$\bullet$ 3 axial symmetries $\vR(\ve_{1},\pi),\vR(\bnn_{k},\pi)$, \\
				with $\bnn_{k}$ from $\DD_{3}$ case\\
				$\bullet$ 4 plane symmetries $\bs(\ve_{3}),\bs(\ve_{2})$ \\
				and $\bs(\mm_{k})$ with\\
				$\mm_{1}=\vR\left(\ve_{3},\frac{\pi}{6}\right)\ve_{1}, 
				\mm_{2}=\vR\left(\ve_{3},\frac{5\pi}{6}\right)\ve_{1}$\\
				\\
			\end{tabular}
			&
			\begin{tabular}{l}
				(10) $\octa^{-}$:\\
				$\bullet$ 3 axial symmetries $\vR(\ve_{i},\pi)$, $i=1,2,3$;\\
				$\bullet$ 6 plane symmetries $\bs(\ve_{j}\pm\ve_{i})$, $i<j$\\
				\hline
				(11) $\OO(2)$:\\
				$\infty$ axial symmetries (axis inside a common plane).\\
				\hline
				(12) $\OO(2)^{-}$:\\
				$\infty$ plane symmetries (all plane with a common axis).\\
			\end{tabular}\\
			\hline
		\multicolumn{2}{|l|}{(13) $\OO(3)$ (degenerate case) : $\infty$ plane and axial symmetries ($\bP=0$).} \\
		\hline
		\end{tabular}
		\caption{Order two symmetries for symmetry classes of the piezoelectricity tensor.}
		\label{tab:Second_Order_Sym_groups}
	\end{table}
\end{center}
From this table, we deduce, for a non null piezoelectricity tensor, two distinct situations:
\noindent \textbf{(a)} there exists at least one symmetry plane or at least two axial symmetries (cases $(3)$ to $(12)$ of \autoref{tab:Second_Order_Sym_groups}). In these 10 cases, the symmetry class is directly characterized by order two symmetries and can be determined using the algorithm depicted in \autoref{fig:Piez_Sym_Class_Second_Order}.
\begin{center}
	\begin{figure}[H]
		\includegraphics[scale=1]{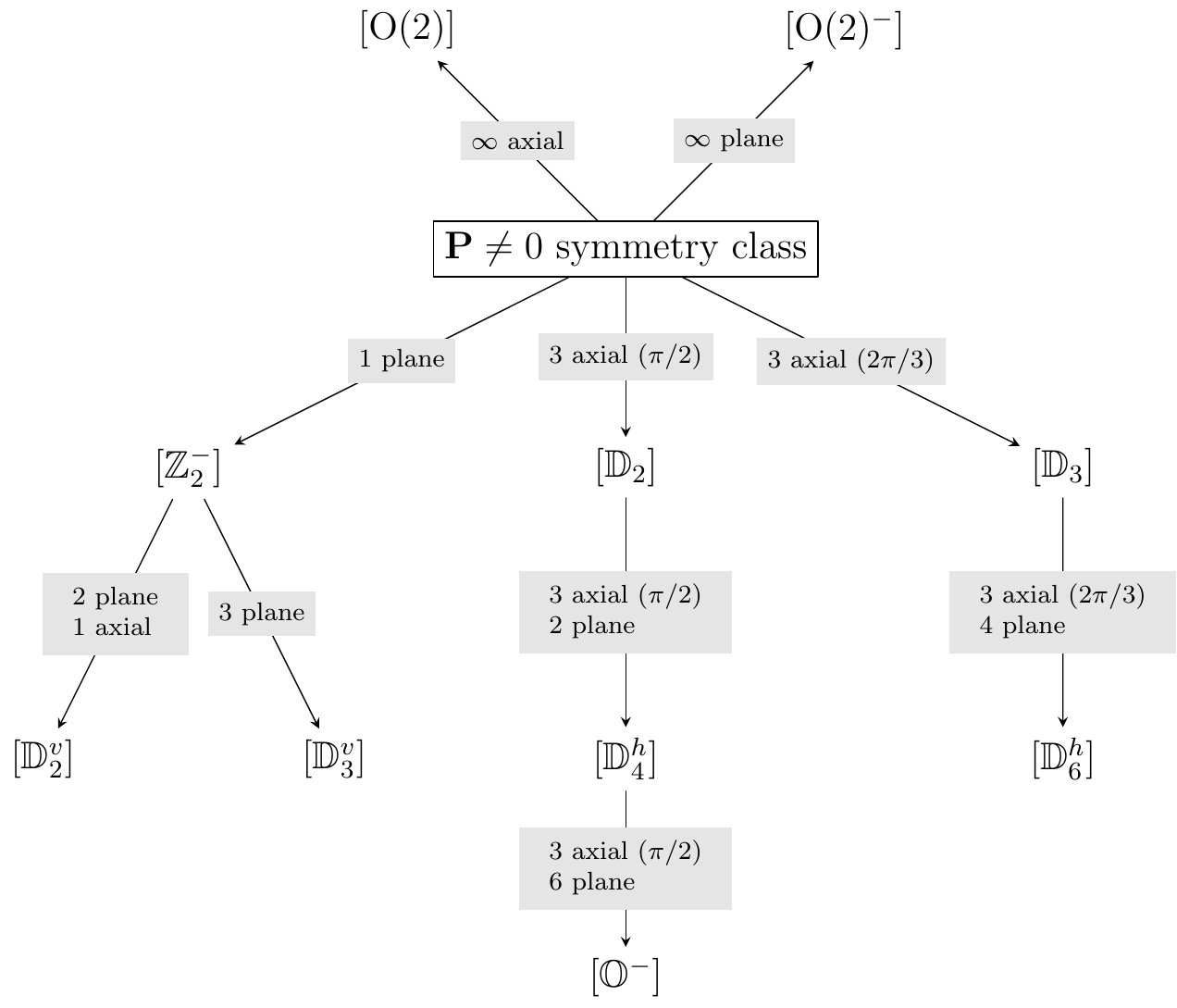}
		\caption{Symmetry classes of a non null piezoelectricity tensor characterized by order two symmetries.}
		\label{fig:Piez_Sym_Class_Second_Order}
	\end{figure}
\end{center}
\noindent \textbf{(b)} there is no plane symmetry in which case, either there is no any order two symmetry, or there exists a unique axial symmetry (cases $(1)$ and $(2)$ of \autoref{tab:Second_Order_Sym_groups}). In these 2 cases, further operations are needed to identify the symmetry class of the tensor.
If we are in the situation $(1)$ of \autoref{tab:Second_Order_Sym_groups}, the symmetry class is identified by the following extra test:
\begin{lem}\label{lem:No_Axe_No_Plan}
	Let suppose that $\bP\in \Pie$ has no any order two symmetry and let us consider the equation
	\begin{equation}\label{eq:Order3_Rotation}
	\vR\left(\bnn,\frac{2\pi}{3}\right)\star \bP=\bP
	\end{equation}
	with unknown unit vector $\bnn=(n_1,n_2,n_3)$ in $\RR^{3}$. If such equation~\eqref{eq:Order3_Rotation} has a solution, then the symmetry class of $\bP$ is $[\ZZ_{3}]$, otherwise it is $[\triv]$.
\end{lem}
While, for the situation $(2)$ of \autoref{tab:Second_Order_Sym_groups}, the symmetry class is identified by the following extra test:
\begin{lem}\label{lem:One_axe}
	Let suppose that $\bP\in \Pie$ has only one axial symmetry $\vR(\bnn,\pi)$ (given by Theorem~\ref{thm:Piezo_Axis_sym}):
	\begin{itemize}
		\item if $\displaystyle{\vR\left(\bnn,\frac{\pi}{2}\right)\star \bP=\bP}$
		then $\bP$ symmetry class is $[\SO(2)]$;
		\item if $\displaystyle{\vR\left(\bnn,\frac{\pi}{2}\right)\star \bP=-\bP}$ then $\bP$ symmetry class is $[\ZZ_{4}^-]$;
		\item otherwise  $\bP$ symmetry class is $[\ZZ_2]$.
	\end{itemize}
\end{lem}
\begin{proof}
	From \autoref{tab:Second_Order_Sym_groups}, we know that the only possible symmetry classes of $\bP$ are $[\ZZ_{2}],[\ZZ_{4}^-]$ or $[\SO(2)]$. Now, we can always suppose that $\bnn=\ve_{3}$, so that we can consider the rotation \begin{equation*}
		g=\vR\left(\ve_{3},\frac{\pi}{2}\right)\in \SO(2)
	\end{equation*} 
	and observe that
	$-g\in \ZZ_{4}^-$ (cf. \autoref{sec:O3_closed_subgroups}). Using $g\star \bP$ or $(-g)\star \bP=-(g\star \bP)$, we thus can decide whether the symmetry group of $\bP$ is $\ZZ_2$, $\ZZ_{4}^-$ or  $\SO(2)$. 
\end{proof}

Results of Lemmas~\ref{lem:No_Axe_No_Plan} and~\ref{lem:One_axe} are summarized in \autoref{fig:Piez_Sym_Class_Second_Order_II}. 

\begin{center}
	\begin{figure}[h]
		\includegraphics[scale=1]{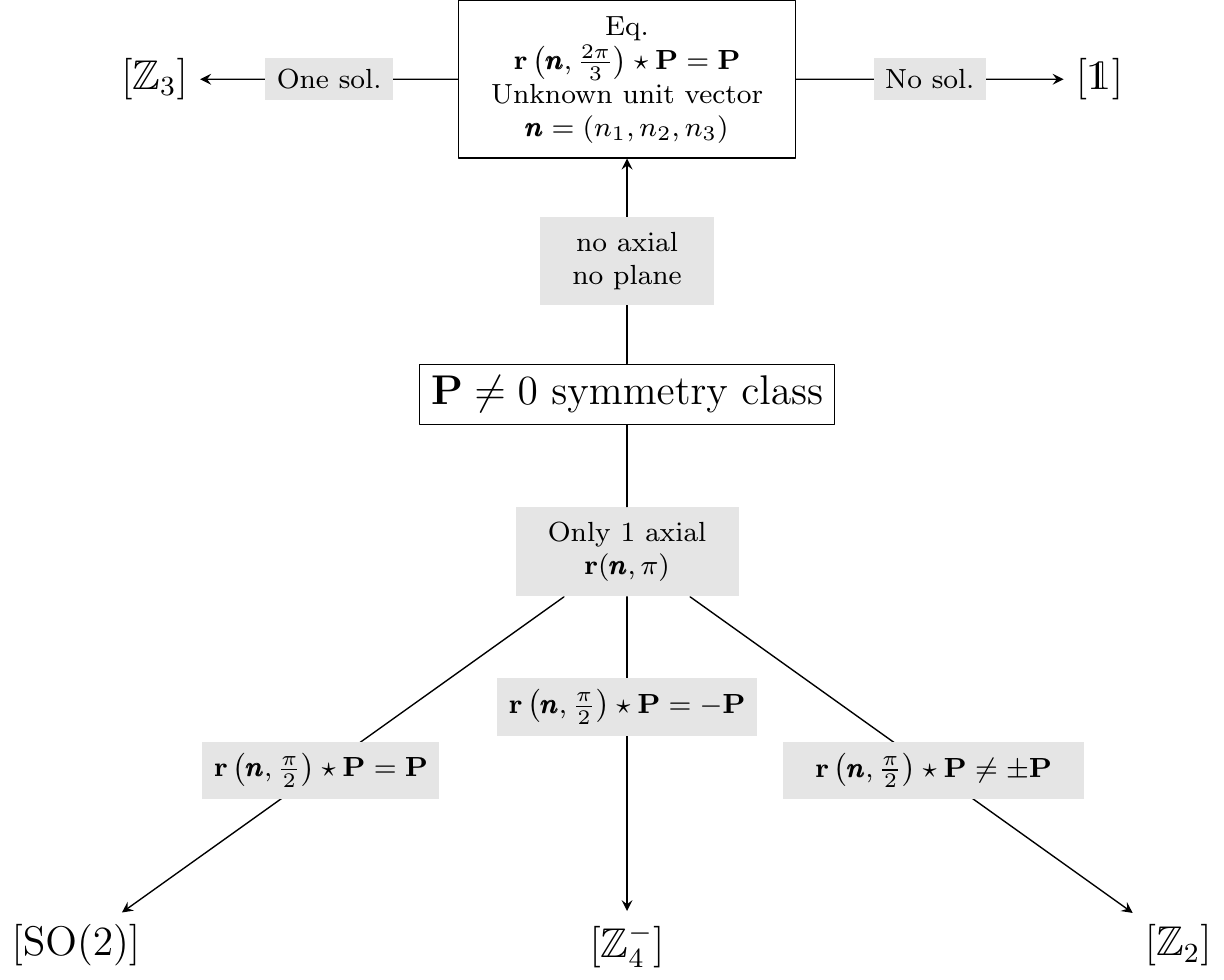}
		\caption{Piezoelectricity symmetry classes from order two symmetries, first cases.}
		\label{fig:Piez_Sym_Class_Second_Order_II}
	\end{figure}
\end{center}

\begin{appendix}

\section{$\OO(3)$-closed subgroups}\label{sec:O3_closed_subgroups}

Classification of $\OO(3)$-closed subgroups is quite classical~\cite{IG1984,Ste1994}. Following Golubitsky and al.~\cite{GSS1988}, $\OO(3)$--closed subgroups can be described using \emph{three types of subgroups}. These subgroups are defined as follows, where $\Idd$ is the neutral element in $\OO(3)$:
\begin{description}
	\item[Type I (Chiral)] A subgroup $\Gamma$ is of type I if it is a subgroup of $\SO(3)$. Type I subgroups are also said to be \emph{chiral} subgroups;
	\item[Type II (Centrosymmetric)] A subgroup $\Gamma$ is of type II if $-\Idd\in \Gamma$. In that case, $\Gamma=K\times \ZZ_{2}^{c}$ where $K$ is some $\SO(3)$ closed subgroup and $\ZZ_{2}^{c}:=\set{\Idd,-\Idd}$. Type II subgroups are also said to be \emph{centrosymmetric};
	\item[Type III] A subgroup $\Gamma$ is of type III if $-\Idd \notin \Gamma$ and $\Gamma$ is not a subgroup of $\SO(3)$.
\end{description}
From a physical point of view, it can be relevant to define subgroups with a privileged axis:
\begin{defn}[Polar subgroups]
	A subgroup $\Gamma$ is said to be \emph{polar} if there exists $\vv\in \RR^{3}$ such that $h\star \vv=\vv$ for all $h\in H$.
\end{defn}

\begin{rem}
	While the Type I-II-III classification leads to \emph{disjoint types} of subgroups, the classification using using Centro--symmetric ($I$), Chiral ($C$) and Polar ($P$) characteristic leads to the following disjoint classification:
	\begin{figure}[H]
		\centering
		\includegraphics[scale=1.2]{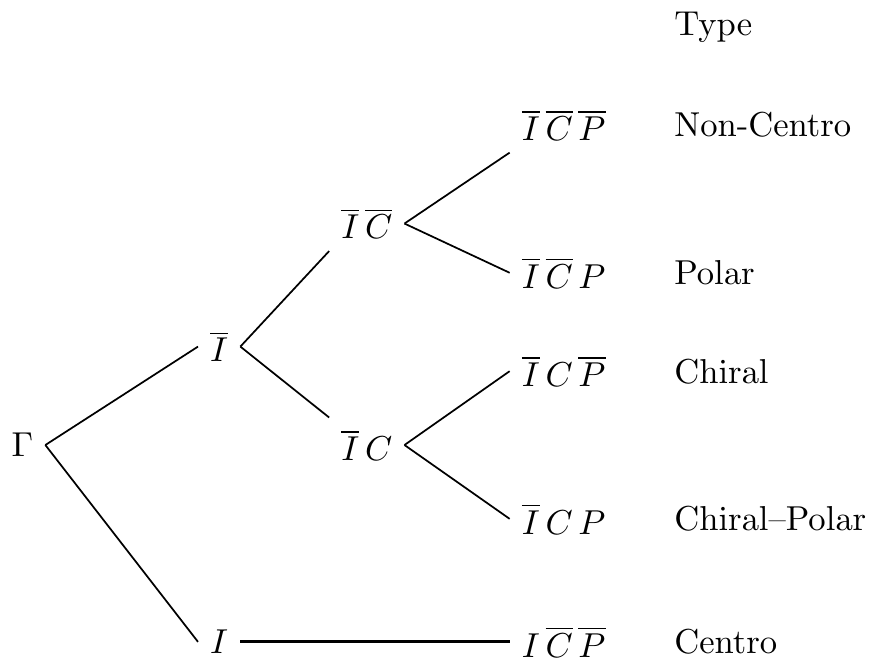}
		\caption{Physical classification of $\OO(3)$-subgroups. $I$, $C$ and $P$
			indicate, respectively, the Centro-symmetric, Chiral and Polar properties of a subgroup $\Gamma$. An overline indicates that the associated property is lacking.}
		\label{fig:StrGrp}
	\end{figure} 
\end{rem}
The following different tables (\autoref{tab:TypeI_Dictionnary},-\ref{tab:TypeII_Dictionnary},-\ref{tab:TypeIII_Dictionnary}) detail the physical  characteristics of the \red{32 crystallographic} point groups of type I, II and III. In these tables, the correspondence between the Golubtisky's notations and the crystallographers' ones such as the Hermann--Maugin or the Schonflies systems are provided. Note also that the column \emph{System} refers to the Bravais lattice, and \emph{Space Groups} indicates the reference of the space groups having this group as its point group \red{(see~\cite{Ste1994} for instance)}.

\begin{table}[H]
	\begin{center}
		\begin{tabular}{|c|>{$}c<{$}|>{$}c<{$}|>{$}c<{$}|c|c|}
			\hline
			\multicolumn{6}{|c|}{Type I closed subgroups}\\
			\hline\hline
			System&\text{Hermann-Maugin} & \displaystyle \text{Schonflies} & \displaystyle \text{Golubitsky}& Nature&Space Groups\\ \hline			
			Triclinic 	&\displaystyle 1 & \displaystyle \ZZ_{1} &\displaystyle \triv  &$CP$&1\\ \hline
			
			Monoclinic 	&	\displaystyle 2 & \displaystyle \ZZ_{2}& \displaystyle \ZZ_{2}  &$CP$&3-5\\ \hline
			Orthotropic	&	\displaystyle 222 & \displaystyle \DD_{2}& \displaystyle \DD_{2}&$C$  &16-24\\ \hline
			Trigonal&\displaystyle 3 & \displaystyle \ZZ_{3}& \displaystyle \ZZ_{3} &$CP$&143-146 \\ \hline
			
			Trigonal&	\displaystyle 32 & \displaystyle \DD_{3}& \displaystyle \DD_{3} &$C$&149-155\\ \hline
			Tetragonal&	\displaystyle 4 & \displaystyle \ZZ_{4}& \displaystyle \ZZ_{4}&$CP$&75-80 \\ \hline
			Tetragonal&	\displaystyle 422 & \displaystyle \DD_{4}& \displaystyle \DD_{4} &$C$&89-98\\ \hline
			Hexagonal&	\displaystyle 6 & \displaystyle \ZZ_{6}& \displaystyle \ZZ_{6} &$CP$&168-173 \\ \hline
			
			Hexagonal&	\displaystyle 622 & \displaystyle \DD_{6}& \displaystyle \DD_{6}&$C$ &177-182 \\ \hline
			&\displaystyle \infty & \displaystyle \ZZ_{\infty}& \displaystyle \SO(2) &$CP$ &\\ \hline
			&\displaystyle \infty 2 & \displaystyle \DD_{\infty}& \displaystyle \OO(2) &$C$ &\\ \hline
			Cubic&\displaystyle 23 & \displaystyle \mathrm{T} & \displaystyle\tetra &$C$&195-199\\ \hline
			Cubic	&\displaystyle 432 & \displaystyle \mathrm{O} & \displaystyle\octa& $C$ &207-214 \\ \hline
			&\displaystyle 532 & \displaystyle \mathrm{I} & \displaystyle\ico &$C$&\\ \hline
			&\displaystyle \infty \infty  & & \displaystyle \SO(3)  &$C$&\\ \hline
		\end{tabular}
	\end{center}		
	\caption{Type I closed subgroups: designation and characteristics ($C$ = Chiral, $P$=Polar).}
	\label{tab:TypeI_Dictionnary}		
\end{table}

\begin{table}[H]
	\begin{center}
		\begin{tabular}{|c|>{$}c<{$}|>{$}c<{$}|>{$}c<{$}|c|c|}
			\hline
			\multicolumn{6}{|c|}{Type II closed subgroups}\\
			\hline\hline
			System&\text{Hermann-Maugin} & \displaystyle \text{Schonflies} & \displaystyle \text{Golubitsky}& Nature&Space Group\\ \hline
			Triclinic&\displaystyle \bar{1} & \displaystyle \ZZ_{i}& \displaystyle \ZZ_{2}^{c}&$I$&2\\ \hline
			Monoclinic&\displaystyle 2/m & \displaystyle \ZZ_{2h}& \displaystyle \ZZ_{2}\times \red{\ZZ_{2}^c} &$I$&10-15\\ \hline
			Orthotropic&\displaystyle mmm & \displaystyle \DD_{2h}& \displaystyle \DD_{2}\times \red{\ZZ_{2}^c} &$I$&47-74 \\ \hline
			Trigonal&\displaystyle \bar{3} & \displaystyle \mathrm{S}_{6},\ \ZZ_{3i}& \displaystyle  \ZZ_{3}\times \red{\ZZ_{2}^c} &$I$&147-148\\ \hline
			Trigonal&\displaystyle \bar{3}m & \displaystyle \DD_{3d}& \displaystyle \DD_{3}\times \red{\ZZ_{2}^c} &$I$&162-167 \\ \hline
			Tetragonal&\displaystyle 4/m & \displaystyle \ZZ_{4h}& \displaystyle \ZZ_{4}\times \ZZ_{2}^{c}&$I$&83-88 \\ \hline
			Tetragonal&\displaystyle 4/mmm & \displaystyle \DD_{4h}& \displaystyle \DD_{4}\times \ZZ_{2}^{c}&$I$&123-142 \\ \hline
			Hexagonal	&\displaystyle 6/m & \displaystyle \ZZ_{6h}& \displaystyle \ZZ_{6}\times \ZZ_{2}^{c}&$I$&175-176 \\ \hline
			Hexagonal	&\displaystyle 6/mmm & \displaystyle \DD_{6h}& \displaystyle \DD_{6}\times \ZZ_{2}^{c}&$I$&191-194 \\ \hline
			&\displaystyle \infty /m& \displaystyle \ZZ_{\infty h}& \displaystyle \SO(2)\times \ZZ_{2}^{c} &$I$&\\ \hline
			&\displaystyle \infty /mm & \displaystyle \DD_{\infty h}& \displaystyle \OO(2)\times \ZZ_{2}^{c} &$I$&\\ \hline		
			Cubic	&\displaystyle m\bar{3} & \displaystyle \mathrm{T}_{h}& \displaystyle\tetra\times \ZZ_{2}^{c} &$I$&200-206\\ \hline
			Cubic	&\displaystyle m\bar{3}m & \displaystyle \mathrm{O}_{h}& \displaystyle\octa\times \ZZ_{2}^{c}&$I$&221-230 \\ \hline
			&\displaystyle \bar{5}\bar{3}m & \displaystyle \mathrm{I}_{h}& \displaystyle\ico\times \ZZ_{2}^{c}&$I$& \\ \hline
			&\displaystyle \infty /m \infty /m & & \displaystyle \OO(3)& $I$& \\ \hline
		\end{tabular}
	\end{center}
		\caption{Type II closed subgroups: designation and characteristics ($I$ = Centrosymetric).}
	\label{tab:TypeII_Dictionnary}	
\end{table}

\begin{table}[H]
	\begin{center}
		\begin{tabular}{|c|>{$}c<{$}|>{$}c<{$}|>{$}c<{$}|c|c|}
			\hline
			\multicolumn{6}{|c|}{Type III closed subgroups}\\
			\hline\hline
			System&\text{Hermann-Maugin} & \displaystyle \text{Schonflies} & \displaystyle \text{Golubitsky}& Nature&Space Groups\\ \hline
			Monocinic&\displaystyle m & \displaystyle  \ZZ_{s}& \displaystyle \ZZ_{2}^{-} &$P$&6-9 \\ \hline
			Orthotropic&\displaystyle 2mm & \displaystyle \ZZ_{2v}& \displaystyle \DD_{2}^{v}&$P$&25-46  \\ \hline
			Trigonal&\displaystyle 3m & \displaystyle \ZZ_{3v}& \displaystyle \DD_{3}^{v} &$P$&156-161\\ \hline
			Tetragonal&\displaystyle \bar{4} & \displaystyle \mathrm{S}_{4}& \displaystyle \ZZ_{4}^{-} &&81-82 \\ \hline
			Tetragonal&\displaystyle 4mm & \displaystyle \ZZ_{4v}& \displaystyle \DD_{4}^{v} &$P$&99-110 \\ \hline
			Tetragonal&\displaystyle \bar{4}2m & \displaystyle \DD_{2d}& \displaystyle \DD_{4}^{h}& &111-122\\ \hline
			Hexagonal&\displaystyle \bar{6} & \displaystyle \ZZ_{3h}& \displaystyle \ZZ_{6}^{-} &&174 \\ \hline
			Hexagonal&\displaystyle 6mm & \displaystyle \ZZ_{6v}& \displaystyle \DD_{6}^{v} &$P$&183-186 \\ \hline
			Hexagonal&\displaystyle \bar{6}2m & \displaystyle \DD_{3h}& \displaystyle \DD_{6}^{h}&&187-190 \\ \hline
			Cubic&\displaystyle \bar{4}3m & \displaystyle \mathrm{T}_{d}& \displaystyle \mathrm{O}^{-}&&215-220 \\ \hline
			&\displaystyle \infty m& \displaystyle \ZZ_{\infty v}& \displaystyle \OO(2)^{-} &$P$& \\ \hline
		\end{tabular}
	\end{center}
	\caption{Type III closed subgroups: designation and characteristics ($P$=Polar).}
	\label{tab:TypeIII_Dictionnary}		
\end{table}	

In the following table a set of generators is detailed for each finite $\OO(3)$-closed subgroups. Those generators are used in \autoref{sec:Norm_Form} to obtain matrix normal forms for elements of $\Pie$.

\begin{table}[H]
\renewcommand{\arraystretch}{2}
		\begin{center}
		\begin{tabular}{|>{$}l<{$}|>{$}c<{$}|>{$}c<{$}|}
			\hline
			\displaystyle \text{Group}& \text{Order} & \displaystyle \text{Generators}\\ \hline
			\displaystyle \ZZ^-_{2} & 2  &\displaystyle \bs(\ve_{3})\\ \hline
			\ZZ_{n},\ n\geq2 & n & \vR\left(\ve_{3};\sfrac{2\pi}{n}\right) \\ \hline
			\displaystyle \DD_{n},\ n\geq2 & 2n & \vR\left(\ve_{3};\sfrac{2\pi}{n}\right),\ \vR(\ve_{1};\pi)\\ \hline
			\displaystyle \ZZ^-_{2n},\ n\geq2   & 2n & \displaystyle -\vR\left(\ve_{3};\sfrac{\pi}{n}\right) \\ \hline
			\displaystyle \DD_{2n}^h,\ n\geq2 & 4n & \displaystyle -\vR\left(\ve_{3};\sfrac{\pi}{n}\right),\ \vR(\ve_{1},\pi)\\ \hline
			\displaystyle \DD_{n}^v,\ n\geq2  & 2n & \displaystyle  \vR\left(\ve_{3};\sfrac{2\pi}{n}\right),\ \bs(\ve_{1}) \\ \hline
			\displaystyle \tetra  & 12 & \displaystyle \vR(\ve_{3};\pi),\ \vR(\ve_{1};\pi),\ \vR(\ve_{1}+\ve_{2}+\ve_{3};\sfrac{2\pi}{3})\\ \hline
			\displaystyle \octa   & 24 & \displaystyle \vR(\ve_{3};\sfrac{\pi}{2}),\ \vR(\ve_{1};\pi),\ \vR(\ve_{1}+\ve_{2}+\ve_{3};\sfrac{2\pi}{3})\\ \hline
			\displaystyle \octa^-  & 24 & \displaystyle -\vR(\ve_{3};\sfrac{\pi}{2}),\ \bs(\ve_{2}-\ve_{3}) \\ \hline 
			\displaystyle \ico & 60 & \vR(\ve_{3};\pi),\ \vR(\ve_{1}+\ve_{2}+\ve_{3};\sfrac{2\pi}{3}), \vR(\ve_{1}+\phi\ve_{3};\sfrac{2\pi}{5})\quad \phi:=(1+\sqrt{5})/2 \\ \hline
		\end{tabular}
	\end{center}
\caption{Generators of \red{finite} closed $\OO(3)$-subgroups}
\label{tab:TabGen}	
\end{table}
To have a geometrical picture of these groups some illustrations are provided.  In~\autoref{fig:IZD}, \autoref{fig:IIZD} and \autoref{fig:III}, figures that are invariant with respect to groups of type I, II and II are provided. On these figures:
\begin{itemize}
	\item the rotational order of the invariance is indicated on the rotation axis (depicted with an arrow);
	\item symmetry planes are indicated in solid lines and without arrow;
	\item arrows drawn on the figures indicate the \emph{spin} of the object. The presence of spin is due to chirality.
\end{itemize}


\begin{figure}[H]
	\centering
	\subfloat[]{\includegraphics[scale=0.4]{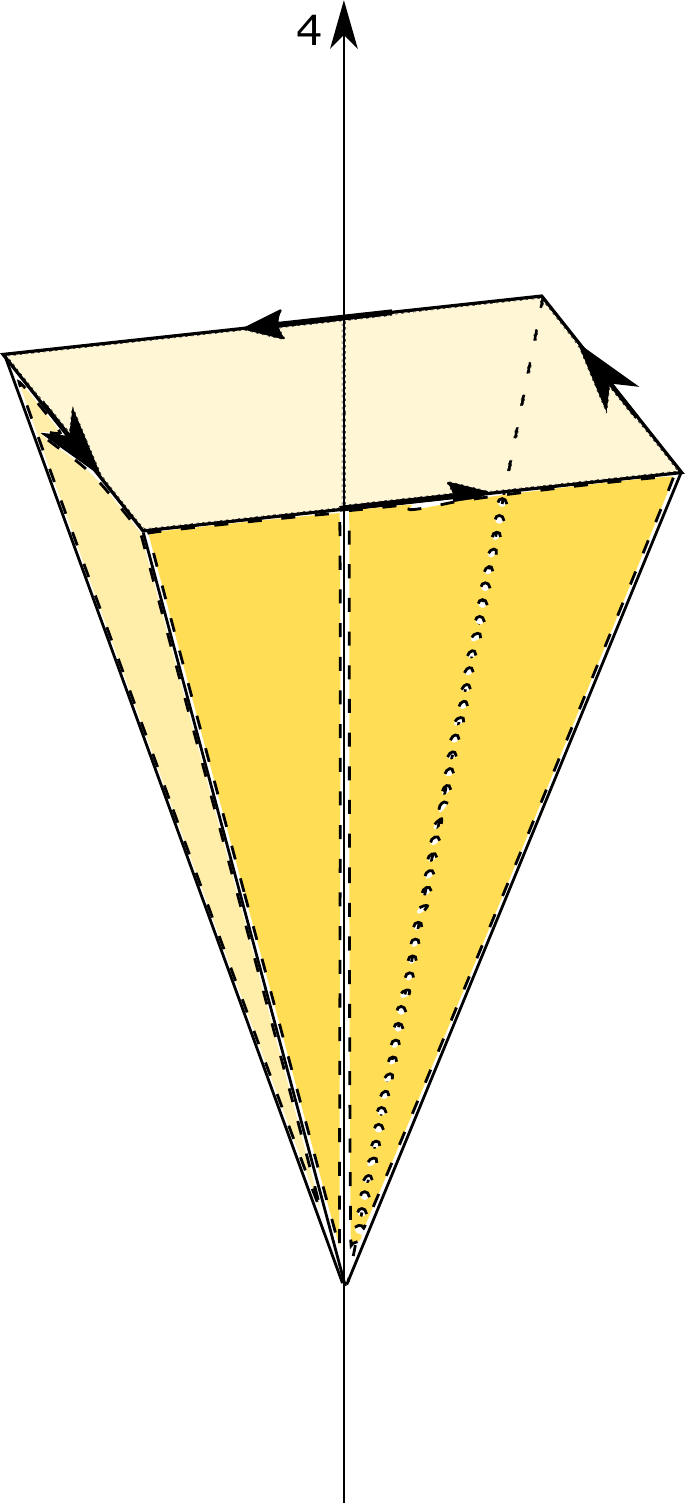}} \hspace{3cm} %
	\subfloat[]{\includegraphics[scale=0.4]{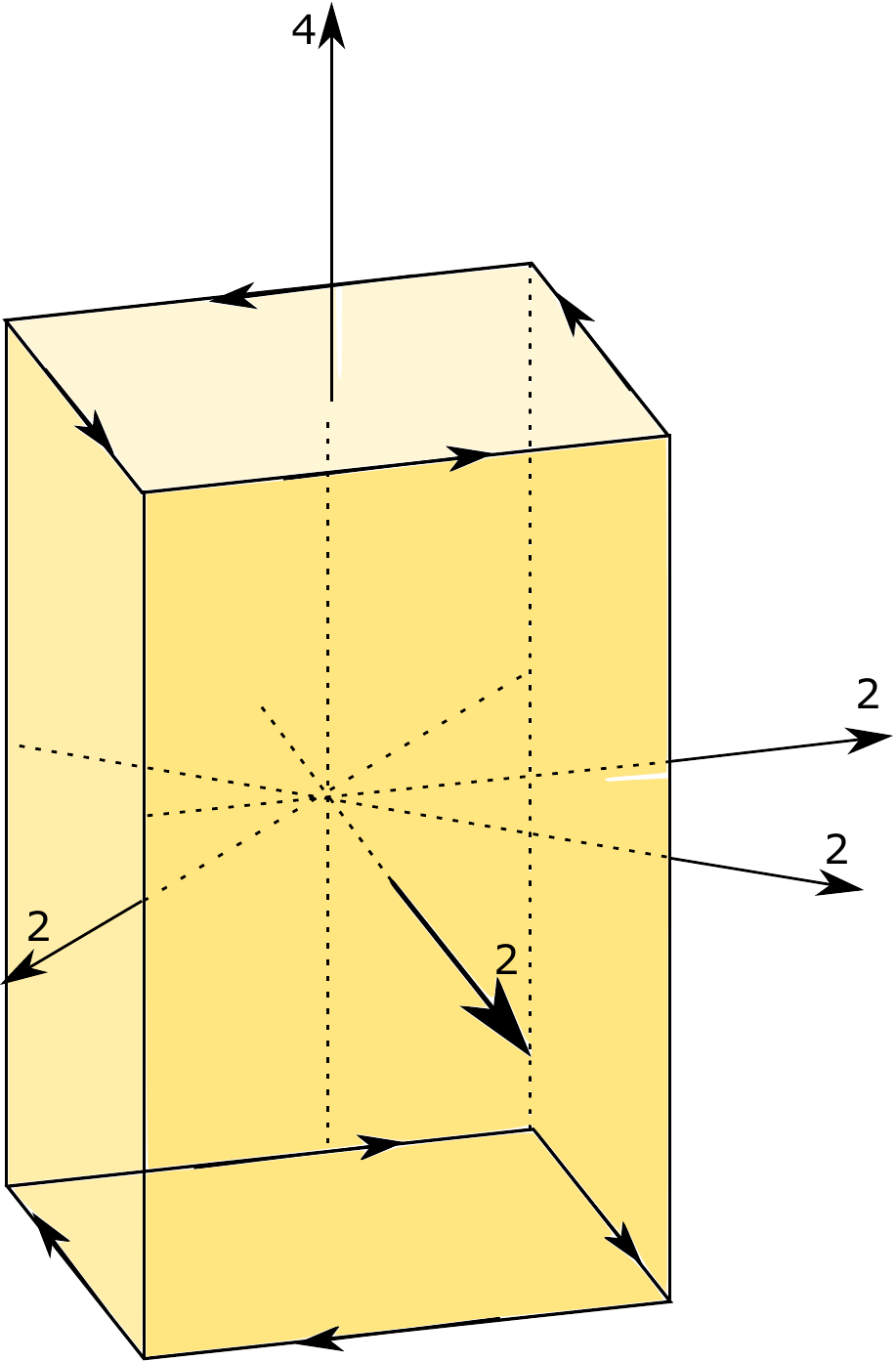}}
	\caption{Type I invariant figures: (A) non regular oriented tetrahedron, $\ZZ_{4}$-invariant (Chiral and Polar), while (B) non cubic twisted rectangular parallelepiped,  $\DD_{4}$-invariant (Chiral).}
	\label{fig:IZD}
\end{figure}

\begin{figure}[H]
	\centering
	\subfloat[]{\includegraphics[scale=0.4]{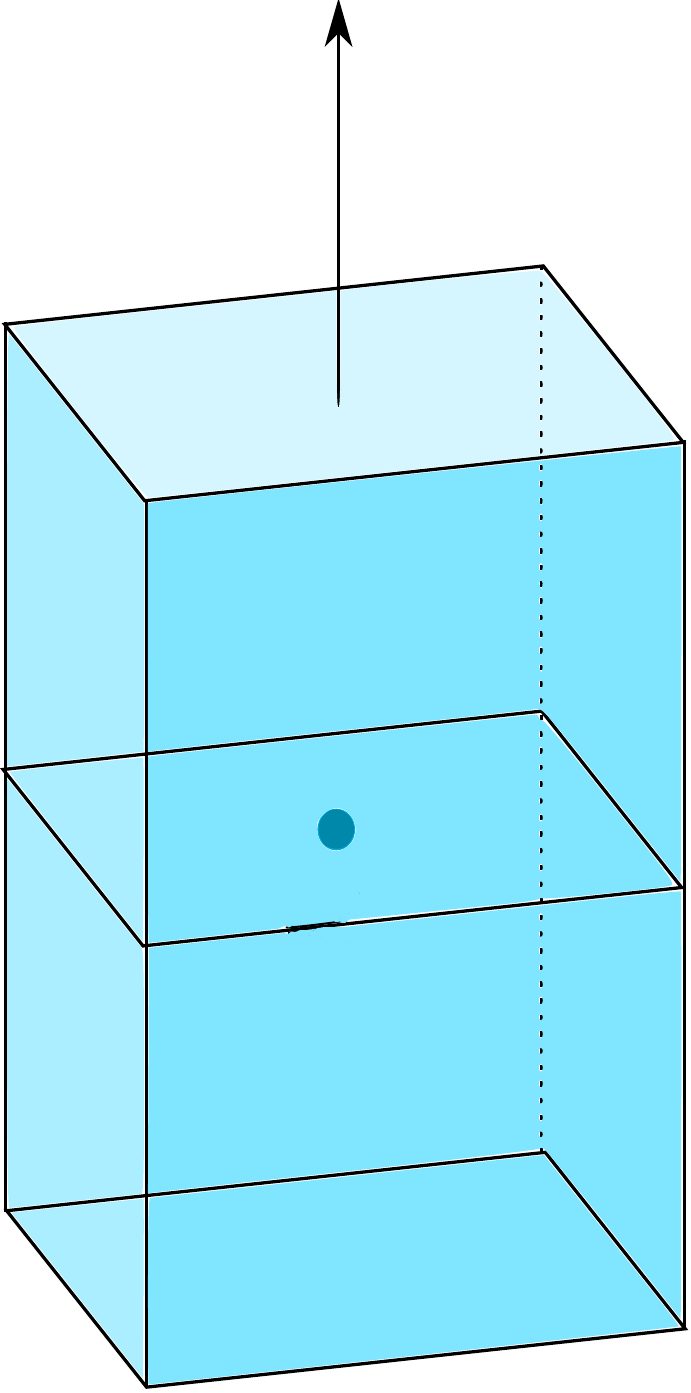}} \hspace{3cm} %
	\subfloat[]{\includegraphics[scale=0.4]{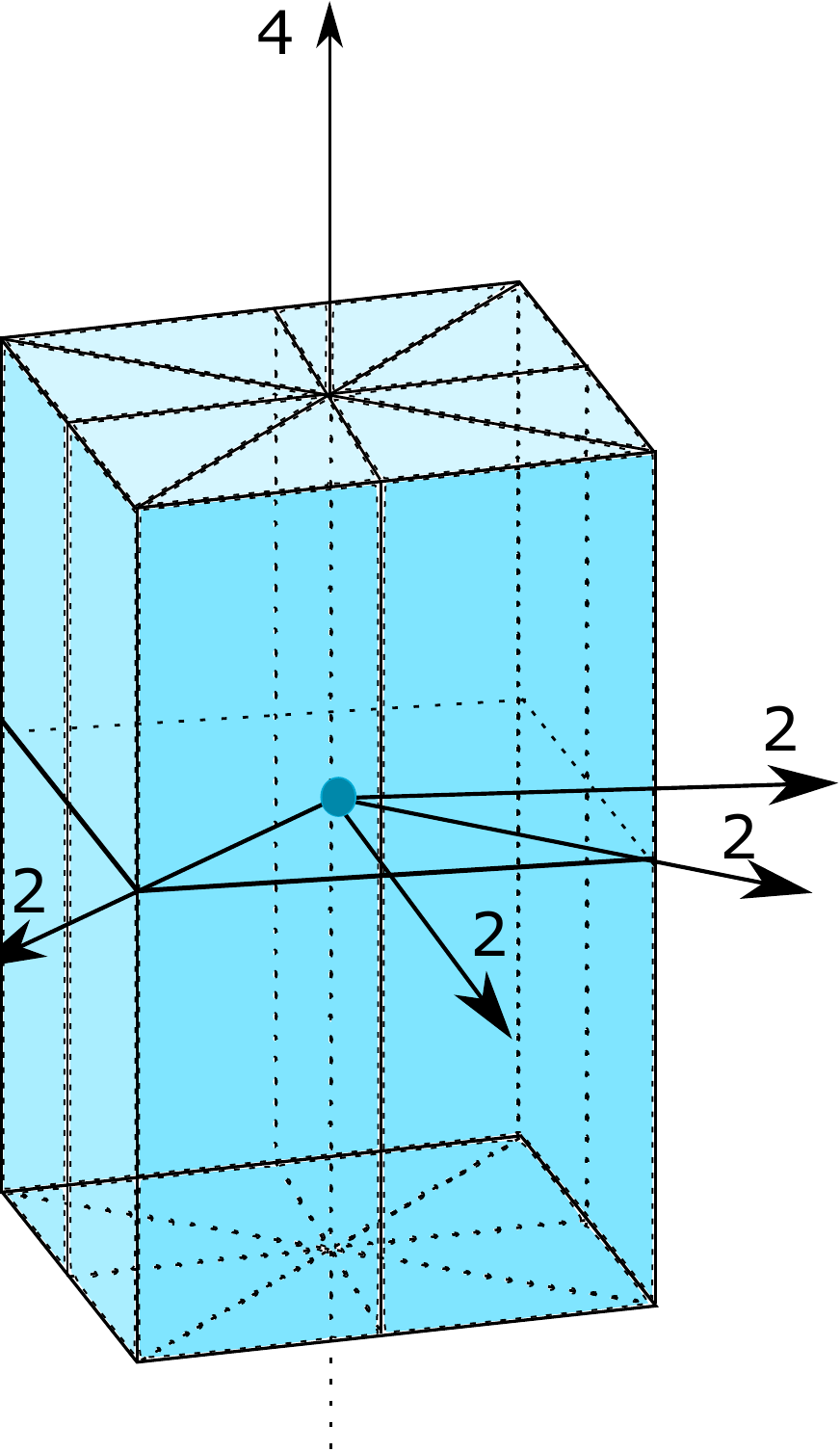}}
	\caption{Different invariant figures of Type II: (A) is $\ZZ_{4}\times\ZZ_2^c$-invariant, while (B) is $\DD_4\times\ZZ_2^c$%
		-invariant. The central inversion is indicated by a dot.}
	\label{fig:IIZD}
\end{figure}

\begin{figure}[H]
	\centering
	\subfloat[]{\label{fig:Z4m}\includegraphics[scale=0.4]{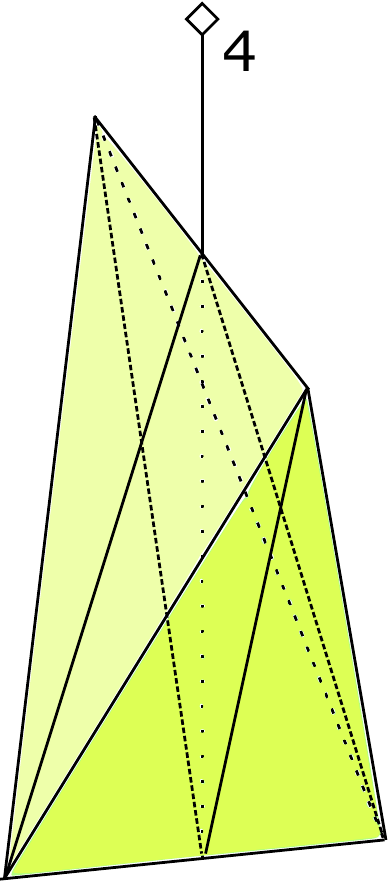}} 
	\hspace{3cm} \subfloat[]{\label{fig:D4v}%
		\includegraphics[scale=0.4]{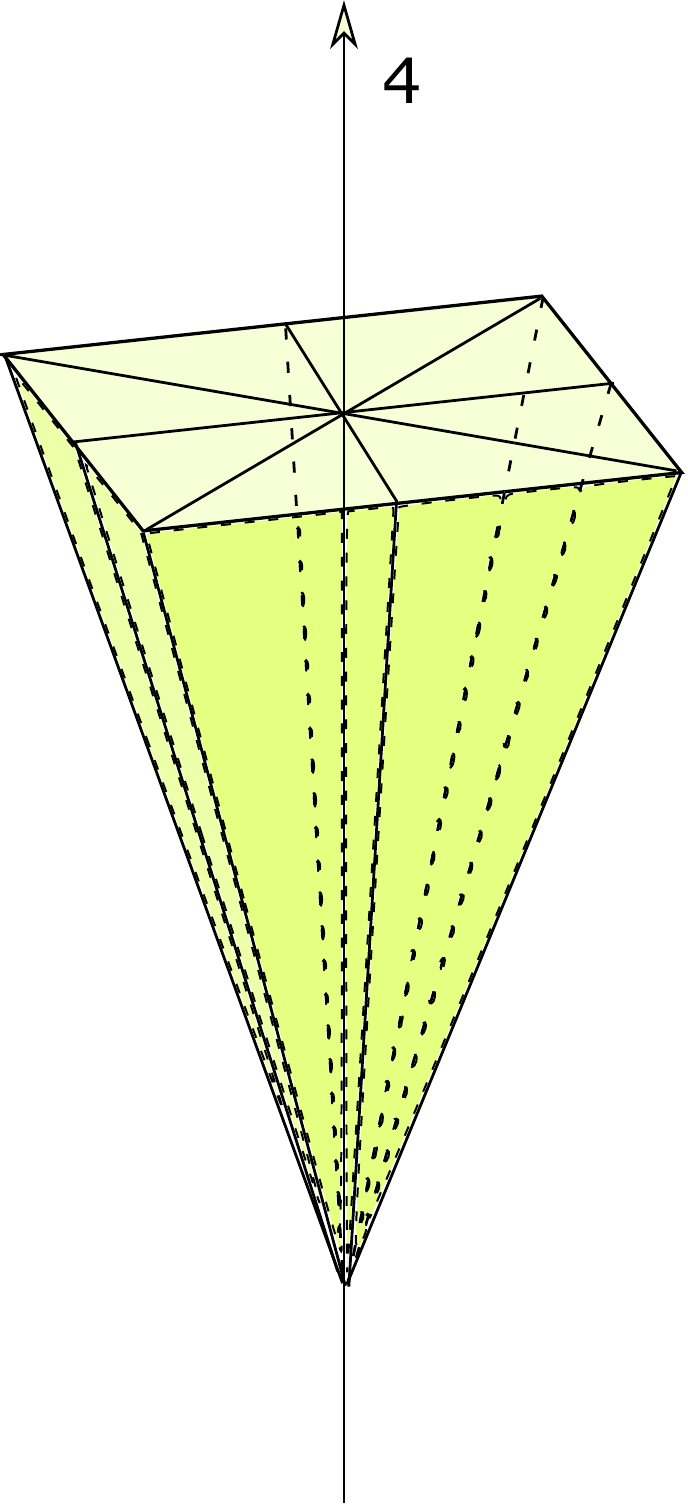}}
	\caption{Type III invariant figures: (A) non regular tetrahedron, $\ZZ_{4}^-$-invariant, while (B) is $\DD_{4}^v$-invariant (Polar). The diamond shape indicates an axis of rotoinversion.}
	\label{fig:III}
\end{figure}
	
\section{Normal forms of piezoelectricity tensors} \label{sec:Norm_Form}

For any piezoelectricity tensor $\bP\in \Pie$ with symmetry class $[H]$, there exists some $g\in \OO(3)$ such that the symmetry group of $\bP_{0}=g\star\bP$ is exactly $H$. The piezoelectricity tensor $\bP_0$ is said to be a \emph{normal form} of $\bP$, which will be given in a \emph{matrix \red{normal} form}. 
To define a matrix representation, let us first observe that any piezoelectric tensor is associated with a linear application from $\Sym^2$ to $\RR^3$, so we can consider its  matrix representation $\MaP$ once the bases of these spaces are given. Let $ \mathcal{B}:=(\ve_{1},\ve_{2},\ve_{3})$  be some orthonormal basis of $\RR^{3}$ and 
\begin{eqnarray*}\label{Basis}
	\ve_{ij}&:=&\left( \frac{1-\delta _{ij}}{\sqrt{2}}+\frac{\delta _{ij}}{2}\right) \left(\mathbf{e}_{i}\otimes \mathbf{e}_{j}+\mathbf{e}_{j}\otimes\mathbf{e}_{i}\right)
\end{eqnarray*}%
be the associated orthonormal basis of $\Sym^2$.  With respect to bases $(\ve_{ij})$ and $(\ve_{i})$, any piezoelectricity tensor $\bP\in\Pie$ can be represented by the matrix
\begin{equation*}
\MaP:=
\begin{pmatrix}
P_{111} & P_{122}  & P_{133}   & \sqrt{2}P_{123}& \sqrt{2}P_{113}&\sqrt{2}P_{112}\\
P_{211} & P_{222}  & P_{233}   & \sqrt{2}P_{223}& \sqrt{2}P_{213}&\sqrt{2}P_{212}\\
P_{311} & P_{322}  & P_{333}   & \sqrt{2}P_{323}& \sqrt{2}P_{313}&\sqrt{2}P_{312}\\
\end{pmatrix}=
\begin{pmatrix}
p_{11} & p_{12} & p_{13}  & p_{14} & p_{15}&p_{16}\\
p_{21} & p_{22} & p_{23}  & p_{24} & p_{25}&p_{26}\\
p_{31} & p_{32} & p_{33} &p_{34} & p_{35}&p_{36}\\
\end{pmatrix}
\end{equation*}

A \emph{matrix normal form} $[\bP_0]$ of a given symmetry class $[H]$ is now obtained from some $\bP_0$ such that $h\star \bP_0=\bP_0$ for all $h\in H$, so that $\bP_0$ belong to the \emph{fixed-point set}
\begin{equation*}
\fix(H):=\set{\bP\in \Pie,\quad h\star\bP=\bP,\quad \forall h\in H} 
\end{equation*}
which is a vector space of dimension given in~\autoref{tab:Dim_Fixed_Point}. In this Table, the dimensions are obtained using the trace formulas detailed in~\autoref{app:Dim_Fixed_Point}.

\begin{table}[H]
	\begin{tabular}{|c|c|c|c|c|}
		\hline
		$[G_{\bP}]$ & $[\triv]$ & $[\ZZ_{2}]$ & $[\ZZ_{3}]$  & $[\SO(2)]$ \\ \hline 
		$\dim(\fix(G_{\bP}))$ & $18$ & $8$ & $6$ & $4$  \\ \hline \hline
		$[G_{\bP}]$ &    & $[\DD^v_{2}]$ & [$\DD^v_{3}]$  & $[\OO(2)^{-}]$  \\ \hline 
		$\dim(\fix(G_{\bP}))$&  & $5$ & $4$ & $3$ \\ \hline \hline
		$[G_{\bP}]$  & $[\ZZ^-_{2}]$ & $[\ZZ^-_{4}]$   &  & \\ \hline 
		$\dim(\fix(G_{\bP}))$ & $10$ & $4$ &  & \\ \hline \hline
		$[G_{\bP}]$ &    & $[\DD^h_{4}]$ & [$\DD^h_{6}]$   &  \\ \hline 
		$\dim(\fix(G_{\bP}))$&  & $2$ & $1$  & \\ \hline \hline
		$[G_{\bP}]$ &    & $[\DD_{2}]$ & [$\DD_{3}]$   & $[\OO(2)]$  \\ \hline 
		$\dim(\fix(G_{\bP}))$&  & $3$ & $2$  & $1$\\ \hline 
		\hline
		$[G_{\bP}]$ & &  & $[\octa^{-}]$  &    $[\OO(3)]$     \\ \hline 
		$\dim(\fix(G_{\bP}))$&  &  & $1$ &$0$ \\ \hline 
	\end{tabular}%
	\caption{Dimension of the fixed-point sets for each  symmetry class of $\Pie$.}
	\label{tab:Dim_Fixed_Point}
\end{table}

A basis of each fixed-point set is computed using generators of the subgroup $H$ from \autoref{tab:TabGen}: if $H$ is generated by some elements $h_i$, then $\fix(H)$ is defined by the linear system 
\begin{equation*}
h_i\star\bP=\bP,\quad \forall i. 
\end{equation*}
In the case of infinite compact subgroup $H$, as a consequence of Herman's theorem~\cite{Her1945,Auf2008a}, generators will be substituted by a rotation of order larger than the tensor's one.
Finally, any basis of $\fix(H)$ produces a matrix normal form of the associated symmetry class $[H]$. In the following, we write $\bP_H$ for an element of $\fix(H)$ and $\MaP_{H}$ for its matrix representation. Note also that the choice of the generators indicated in \autoref{tab:TabGen} has been made in order to have the following relation, which is obtained by direct computation:
\begin{equation*}
     \Pie=\fix(\DD_2)\oplus\fix(\DD_2^v)\oplus\fix(\ZZ^{-}_{2}).
\end{equation*}
This relation means that any triclinic tensors $\bP$ is the sum of three tensors $\bP$ with higher symmetry classes. The consequences of this decomposition are used  to provide matrix representations in a condensed form. Similar relations are obtained for some symmetry classes in order to simplify computations of normal forms.  

\subsubsection*{Symmetry classes $[\ZZ_{2}],[\DD_{2}],[\DD^{v}_{2}]$}

As discussed in \autoref{sec:O3_closed_subgroups}, classes $[\DD_{n}]$ are chiral, classes $[\DD^{v}_{n}]$ are polar while classes $[\ZZ_{n}]$ have these two features. Using generators from \autoref{tab:TabGen} we have:
\begin{equation*}
	\fix(\ZZ_n)=\fix(\DD_n)\oplus\fix(\DD_n^v)\: 
\end{equation*}

So only the matrices of $\left[\bP\right]_{\DD_{n}}$ and $\left[\bP\right]_{\DD^{v}_{n}}$ will be detailed, the remaining one being deduced from the above relation.

\small{
\begin{equation*}
\MaP_{\DD_{2}}=%
\begin{pmatrix}
0 & 0 & 0 &p_{14}&0&0\\
0 & 0 & 0 &0 &  p_{25} &0\\
0 & 0 & 0 &0 & 0 &p_{36}\\
\end{pmatrix}
\  ,\ 
\MaP_{\DD^{v}_{2}}=
\begin{pmatrix}
0 & 0 & 0 &0&p_{15}&0\\
0 & 0 & 0 &p_{24} &  0 &0\\
p_{31} & p_{32} & p_{33} &0 & 0 &0\\
\end{pmatrix},
\end{equation*}
\begin{equation*}
	\MaP_{\ZZ_{2}}=\MaP_{\DD_{2}}+\MaP_{\DD^{v}_{2}}
\end{equation*}}

\subsubsection*{Symmetry classes $[\ZZ_{3}],[\DD_{3}],[\DD^{v}_{3}]$}

\small{
\begin{equation*}
\MaP_{\DD_{3}}=%
\begin{pmatrix}
p_{11}& -p_{11} & 0 &p_{14}&0&0\\
0 & 0 & 0 &0 &  -p_{14} &-\sqrt{2}p_{11} \\
0 & 0 & 0 &0 & 0 & 0 \\
\end{pmatrix}
\  ,\ 
\MaP_{\DD^{v}_{3}}=
\begin{pmatrix}
0 & 0 & 0 &0&p_{24}&\sqrt{2}p_{21}\\
p_{21} & -p_{21} & 0 &p_{24} &  0 &0\\
p_{31} & p_{31}& p_{33} &0 & 0 &0\\
\end{pmatrix},
\end{equation*}
\begin{equation*}
	\MaP_{\ZZ_{3}}=\MaP_{\DD_{3}}+\MaP_{\DD^{v}_{3}}
\end{equation*}}

\subsubsection*{Symmetry classes $[\SO(2)],[\OO(2)],[\OO(2)^{-}]$}

\small{
\begin{equation*}
\MaP_{\OO(2)}=%
\begin{pmatrix}
0 & 0 & 0  &p_{14}&0&0\\
0 & 0 & 0 &0 &  -p_{14} &0 \\
0 & 0 & 0 &0 & 0 & 0 \\
\end{pmatrix}
\  ,\ 
\MaP_{\OO(2)^{-}}=
\begin{pmatrix}
0 & 0 & 0 &0&p_{24}&0\\
0 & 0 & 0  &p_{24} &  0 &0\\
p_{31} & p_{31}& p_{33} &0 & 0 &0\\
\end{pmatrix}
\end{equation*}
\begin{equation*}
	\MaP_{\SO(2)}=\MaP_{\OO(2)}+\MaP_{\OO(2)^{-}}
\end{equation*}}

\subsubsection*{Symmetry Classes $[\ZZ^{-}_{2}],[\ZZ^{-}_{4}],[\DD^{h}_{4}]$}

\small{
\begin{equation*}
\MaP_{\ZZ^{-}_{2}}=%
\begin{pmatrix}
p_{11} & p_{12} & p_{13}  &0&0&p_{16}\\
p_{21} & p_{22} & p_{23}  &0&0&p_{26}\\
0 & 0 & 0 &p_{34} & p_{35}& 0 \\
\end{pmatrix}
\  ,\ 
\MaP_{\ZZ^{-}_{4}}=%
\begin{pmatrix}
0 & 0 & 0 &p_{14}&-p_{24}&0\\
0 & 0 & 0  &p_{24} &  p_{14} &0\\
p_{31} & -p_{31}& 0 &0 & 0 &p_{36}\\
\end{pmatrix}
\end{equation*}
\begin{equation*}
\MaP_{\DD^{h}_{4}}=
\begin{pmatrix}
0 & 0 & 0 &p_{14}&0&0\\
0 & 0 & 0  &0&  p_{14} &0\\
0 & 0& 0 &0 & 0 &p_{36}\\
\end{pmatrix}
\end{equation*}}

\subsubsection*{Symmetry Classes $[\DD^{h}_{6}], [\octa^{-}]$}

\small{
\begin{equation*}
\MaP_{\DD^{h}_{6}}=
\begin{pmatrix}
p_{11} & -p_{11} & 0  &0&0&0\\
0&0&0 &0&0&-\sqrt{2}p_{11} \\
0 & 0 & 0 &0 & 0 & 0\\
\end{pmatrix}
\  ,\ 
\MaP_{\octa^{-}}=
\begin{pmatrix}
0 & 0 & 0 &p_{14}&0&0\\
0 & 0 & 0  &0 &  p_{14} &0\\
0 & 0 & 0  &0 &  0&p_{14} \\
\end{pmatrix}
\end{equation*}}

\begin{rem}\label{rem:Z6-form}
It is worth noting that a $\ZZ^{-}_{6}$-invariant piezoelectricity tensor can be determined:
\begin{equation*}
\MaP_{\ZZ^{-}_{6}}=
\begin{pmatrix}
p_{11} & -p_{11} & 0  &0&0&\sqrt{2}p_{21}\\
p_{21}&-p_{21}&0 &0&0&-\sqrt{2}p_{11} \\
0 & 0 & 0 &0 & 0 & 0\\
\end{pmatrix}\in  \fix(\ZZ_{6}^-)
\end{equation*}
but is not a normal form for $[\ZZ^{-}_{6}]$ since this class is empty. Indeed, taking $\bP^{\star}:=\vR\left(\ve_{3};\theta\right)\star \bP$ gives the matrix:
\begin{equation*}
[\bP^{\star}]_{\ZZ^{-}_{6}}=
\begin{pmatrix}
p^\star_{11} & -p^\star_{11} & 0  &0&0&\sqrt{2}p^\star_{21}\\
p^\star_{21}&-p^\star_{21}&0 &0&0&-\sqrt{2}p^\star_{11} \\
0 & 0 & 0 &0 & 0 & 0\\
\end{pmatrix}
,\quad\text{with}\quad
\begin{cases}
p^\star_{11}=p_{11}\cos(3\theta) - p_{21} \sin(3\theta)\\
p^\star_{21}=p_{21}\cos(3\theta) +p_{11} \sin(3\theta)
\end{cases}
\end{equation*}
Hence $p^\star_{21} =0\  \text{for}\ \theta=-\frac{1}{3}\mathrm{arctan}\left(\frac{p_{21}}{p_{11}}\right)$
showing that there is always a rotation such that
\begin{equation*}
\vR\left(\ve_{3};\theta\right)\star \bP_{\ZZ^{-}_{6}}\in \fix(\DD_6^h).
\end{equation*} 
\end{rem}

\section{Fixed point set dimensions}\label{app:Dim_Fixed_Point}

Consider $(V,\rho)$  a finite dimensional linear representation of $\OO(3)$ which admits the following harmonic decomposition : 
\begin{equation*}
V\simeq \bigoplus_{k=0}^n \alpha_k\mathbb{W}^k\quad\text{with}\  \alpha_k\mathbb{W}^k=\begin{cases}\emptyset,\ \text{if}\ \alpha_k=0\\
\bigoplus_{i=1}^{\alpha_k} \mathbb{W}^k,\ \text{otherwise}
\end{cases}
\end{equation*}
in which $\mathbb{W}^k$ is either  $\HH^k$ or $\HH^{\sharp k}$. Then for the Type I closed-subgroups we have the following formulas:

\begin{equation}
\begin{split}
\text{dim}\left(\fix (\ZZ_{p})\right)& = 2\sum_{k=0}^{n} \alpha_{k}\floor*{\frac{k}{p}}+\sum_{k=0}^{n} \alpha_{k}\\
\text{dim}\left(\fix (\DD_{p})\right)& = \sum_{k=0}^{n} \alpha_{k}\floor*{\frac{k}{p}}+\sum_{k=0}^{\floor*{\frac{n}{2}}}\alpha_{2k}\\
\text{dim}\left(\fix (\tetra) \right)& = \sum_{k=0}^{n} \alpha_{k}\left(2\floor*{\frac{k}{3}}+\floor*{\frac{k}{2}}-k+1\right)\\
\text{dim}\left(\fix (\octa) \right)& = \sum_{k=0}^{n} \alpha_{k}\left(\floor*{\frac{k}{4}}+\floor*{\frac{k}{3}}+\floor*{\frac{k}{2}}-k+1\right)\\
\text{dim}\left(\fix (\ico) \right)& = \sum_{k=0}^{n} \alpha_{k}\left(\floor*{\frac{k}{5}}+\floor*{\frac{k}{3}}+\floor*{\frac{k}{2}}-k+1\right)\\
\end{split}
\end{equation}
In these expressions $\floor*{\cdot}$ denotes the classical floor function. The dimensions for $\fix(\SO(2))$ and $\fix(\OO(2))$ are respectively obtained from the $\text{dim}\left(\fix (\ZZ_{p})\right)$ and $\text{dim}\left(\fix (\ZZ_{p})\right)$ by considering $p>n$, i.e.:
\ben
\text{dim}\left(\fix (\SO(2))\right) = \sum_{k=0}^{n} \alpha_{k}\quad;\quad
\text{dim}\left(\fix (\OO(2))\right) =\sum_{k=0}^{\floor*{\frac{n}{2}}}\alpha_{2k}.
\een

For the  Type III closed-subgroups, the following formulas are 
\begin{align*}
\text{dim}\left(\fix (\ZZ_{2p}^-)\right)&=2\sum_{k=0}^n \alpha_k \floor*{\frac{k+p}{2p}}   \\ \notag
\text{dim}\left(\fix (\DD_{p}^v)\right)&=\sum_{k=0}^{n} \alpha_{k} \floor*{\frac{k}{p}}+ \sum_{k=0}^{\floor*{\frac{n-1}{2}}} \alpha_{2k+1}\\
\text{dim}\left(\fix (\DD_{2p}^h)\right)&=\sum_{k=0}^n \alpha_k \floor*{\frac{k+p}{2p}}   \\ \notag
\text{dim}\left(\fix (\octa^-)\right)&=\sum_{k=0}^n \alpha_k \left(\floor*{\frac{k}{3}}-\floor*{\frac{k}{4}}\right)   \\ \notag
\end{align*}
The dimension for  $\fix(\OO(2)^{-})$ is  obtained from the $\text{dim}\left(\fix (\DD_{p}^v)\right)$ by considering $p>n$, i.e.:
\begin{align*}
\text{dim}\left(\fix (\OO(2)^-)\right)& =\sum_{k=0}^{\floor*{\frac{n}{2}}}\alpha_{2k}\\
\end{align*}

\end{appendix}

\end{document}